\documentclass{article}
\usepackage{amsmath}
\usepackage{amsthm}
\usepackage{amssymb}
\usepackage[backend=biber,style=numeric,sorting=none,isbn=false,url=false,eprint=false,citestyle=numeric-comp,backref=false,giveninits]{biblatex}

\usepackage{booktabs} % For formal tables
\usepackage{cleveref}
\usepackage{comment} % comment environment \begin{comment}...
\usepackage{eso-pic} % watermark on pages
\usepackage{graphicx}% watermark on pages
\usepackage{ifthen}
\usepackage[utf8]{inputenc}
\usepackage{multicol}
\usepackage{multirow}
\usepackage{numprint}
\usepackage{tabularx} % Paket für Tabellen mit anpassbarer Breite
\usepackage{tikz} % for plots
\usepackage[textsize=tiny,obeyFinal]{todonotes}
\usepackage{pgfplots} % load after tikz
\usepackage{xifthen}
\usepackage{xcolor} % watermark on pages
\usepackage{xspace}
\usepackage[outline]{contour}

\usetikzlibrary{fit}
\usetikzlibrary{matrix}
\usetikzlibrary{shapes}
\usetikzlibrary{positioning}
\usetikzlibrary{calc,patterns}
\pgfplotsset{compat=1.17}
\tikzstyle{state} = [rectangle, rounded corners, minimum height=1cm,text centered, draw=black, fill=white, text width=2.8cm]

\newtheorem{definition}{Definition}[section]
\newtheorem{corollary}[definition]{Corollary}

\newtheorem{lemma}[definition]{Lemma}
\newtheorem{theorem}[definition]{Theorem}

\newcommand{\ite}{i.\,e.\xspace}

\newcommand{\dt}{\,\mathrm{d}t}
\newcommand{\ds}{\,\mathrm{d}s}
\newcommand{\R}{\mathbb{R}\xspace}

\makeatletter
\newcommand{\@minipagerestore}{\setlength{\parskip}{\medskipamount}}
\makeatother

\newif\ifblackwhite
\newif\ifcolor
\colortrue
\ifcolor
    
    \def\blue{blue}
    
    \def\orange{orange}
    \blackwhitefalse
\else
    
    \def\blue{black}
    
    \def\orange{black}
    \blackwhitetrue
\fi

\bibliography{references}

\title{An improved bound for the ground state energy of a Schrödinger operator on a loop}

\author{Helmut~Linde}

\date{\today}

\begin{document}

\maketitle

\begin{abstract}
    Consider a closed curve of length $2\pi$ with curvature $\kappa(s)$ and the Schrödinger operator $H$ with $\kappa^2$ as the potential term. Let $\lambda_\Gamma$ be the lowest eigenvalue of $H$. The Ovals Conjecture proposed by Benguria and Loss states that $\lambda_\Gamma \ge 1$. While the conjecture remains open, the present work establishes a new lower bound of $0.81$ on $\lambda_\Gamma$, improving on the best previous estimate of approximately $0.6085$. 
\end{abstract}

%==========================================
\section{Introduction}
%==========================================

Let $\Gamma$ be a closed curve of length $2\pi$ in $\R^n$ with the curvature $\kappa(s)$ as a function of the arc length.
Let $H_\Gamma$ be the Schr\"odinger operator
\begin{equation} \label{eq:HGamma}
    H_\Gamma = - \Delta  + \kappa^2(s) \quad \mathrm{in} \quad L^2([0,2\pi))
\end{equation}
with periodic boundary conditions and $\lambda_\Gamma$ its lowest eigenvalue. The topic of the present article is the `Ovals Conjecture' which states that this eigenvalue cannot be smaller than one.

The question about the minimal possible value of $\lambda_\Gamma$ has attracted attention since the year 2004, when Benguria and Loss established a connection between this problem and the Lieb-Thirring conjecture in one dimension \cite{BenguriaLoss:Connection}.
They also conjectured that $\lambda_\Gamma \ge 1$ for any $\Gamma$. The class ${\cal F}$ of putative minimizers of $\lambda_\Gamma$ contains the circle and certain point-symmetric oval loops. It can be extended to the edge case of a straight segment of length $\pi$, traversed back and forth, that becomes admissible in a more generalized setting. For all curves in ${\cal F}$ the equality $\lambda_\Gamma = 1$ holds, but it has not yet been proved that this is actually the smallest possible value of $\lambda_\Gamma$.

Several other authors have made contributions towards proving the conjecture of Benguria and Loss: Burchard and Thomas have demonstrated \cite{BurchardThomas:Isoperimetric} that the curves in ${\cal F}$ minimize $\lambda_\Gamma$ locally, \ite, there is no small variation around these curves that reduces $\lambda_\Gamma$.
Denzler \cite{Denzler:Existence} has shown that closed curves minimizing $\lambda_\Gamma$ exist and that these minimizers are planar, smooth, strictly convex curves. Throughout this paper we can therefore assume that $\Gamma$ has these properties without losing generality.

The highest global bound on the eigenvalue so far was established in the year 2006, showing that $\lambda_\Gamma > 0.6085$ for all possible curves that meet the specifications of the conjecture \cite{Linde:ALowerBound}. 
In the present article, that bound on $\lambda_\Gamma$ is improved to $0.81$.

%==========================================
\section{Result statement} \label{sec:Result}
%==========================================

The main result of this article is:
\begin{theorem} \label{theorem:main_results}
    Let $\Gamma$ be a smooth, strictly convex, closed curve of length $2\pi$ in the plane with curvature $\kappa(s)$ and $H_\Gamma$ the Schr\"odinger operator (\ref{eq:HGamma}).
    Then its lowest eigenvalue $\lambda_\Gamma$ satisfies the bound
    \begin{equation} \label{eq:main_result}
        \lambda_\Gamma > \inf_{\substack{\tilde{\nu} \in [0,1] \\ \Delta \in (0, \frac\pi2)}} \max(B_1(\tilde\nu, \Delta), B_2(\tilde\nu, \Delta)),
    \end{equation}
    where 
    \begin{align}
        B_1(\tilde\nu, \Delta) &= \left(1 + 2\tilde\nu G(\Delta)\right)^{-2}, \label{eq:def_B1}\\
        B_2(\tilde\nu, \Delta) &= 1 -  \left(2 - \sec^2 \frac\Delta 2\right)\left( 1-\left(2+2\tilde\nu G(\Delta)-\tilde\nu\right)^{-2}\right) \label{eq:def_B2}
    \end{align}
    with
    \begin{equation} \label{eq:def_G}
        G(\Delta) :=\left(1+\cot\frac\Delta 4\right)^{-2}.
    \end{equation}
\end{theorem}

Despite its involved form, the argument on the r.h.s. of (\ref{eq:main_result}) ultimately represents a combination of elementary functions on the rectangular domain $[0,1]\times(0,\frac\pi2)$.
It is straightforward to evaluate the expression (\ref{eq:main_result}) with a computer and the result is 
\begin{equation}
    \lambda_\Gamma > 0.8246 \text{ (numerically).}
\end{equation}
Also see Figure \ref{fig:contour} for a visual representation of $B_1$ and $B_2$.

\begin{figure}[htbp]
  \centering
  \includegraphics[width=\linewidth]{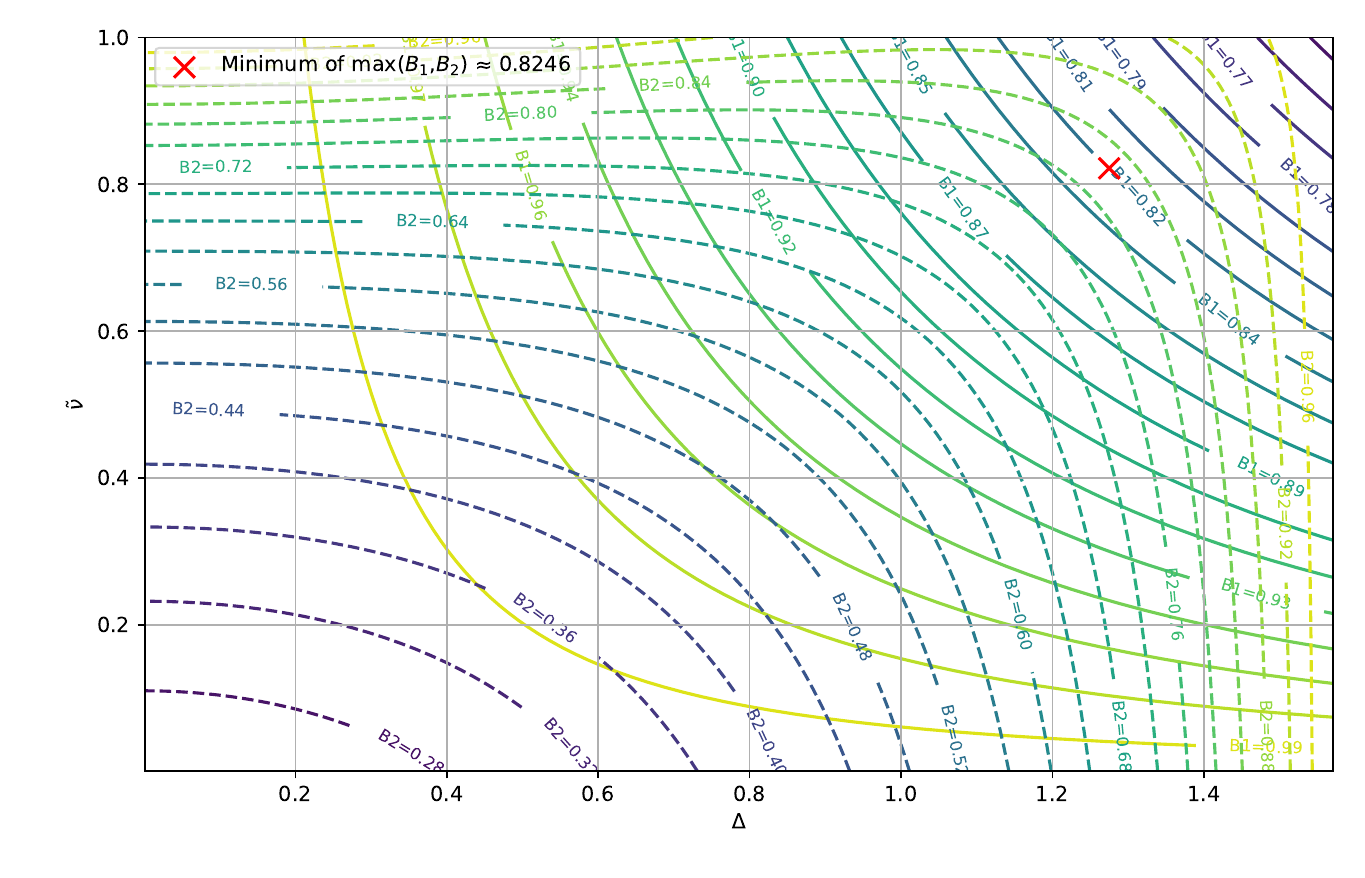}
  \caption{Superposition of contour plots of $B_1$ (solid lines) and $B_2$ (dashed lines). The X marks the point where $\max(B_1,B_2)$ is minimal.}
  \label{fig:contour}
\end{figure}
Alternatively, one can derive analytical bounds on $\lambda_\Gamma$ from Theorem \ref{theorem:main_results} to avoid the use of a computer. 
The function $B_1(\tilde\nu, \Delta)$ decreases strictly mono\-tonously with $\tilde\nu$ and with $\Delta$. Therefore, a crude but simple bound on $\lambda_\Gamma$ can be given based on $B_1$ alone:
\begin{align*}
    \lambda_\Gamma &> \inf_{\substack{\tilde{\nu} \in [0,1] \\ \Delta \in (0, \frac\pi2)}} B_1(\tilde\nu, \Delta) = \lim_{\Delta \rightarrow \frac\pi2} B_1(1, \Delta) \\
    &= \frac{3+2\sqrt 2}{8} \approx 0.7285 
\end{align*}
To obtain better estimates leveraging $B_1$ and $B_2$ jointly, there is a trade-off to be made between accuracy and simplicity of the analytical treatment. The following theorem and its proof are one example for such a trade-off:
\begin{theorem} \label{theorem:main_result2}
    The inequality (\ref{eq:main_result}) implies the weaker bound
    \begin{equation} \label{eq:lambda_greater_0.81}
        \lambda_\Gamma > 0.81
    \end{equation}
\end{theorem}
The rest of the article is devoted to proving Theorem \ref{theorem:main_results} and Theorem \ref{theorem:main_result2}: In Section \ref{sec:Preliminaries}, some groundwork is laid by summarizing results from previous work, introducing the notation, and proving several lemmata. The main part of the proof of Theorem \ref{theorem:main_results} is given in Section \ref{sec:Proof} and this is where new methods of treating the Ovals Conjecture are developed. Section \ref{sec:Proof2} contains the purely technical proof of Theorem \ref{theorem:main_result2}, which is given here for completeness and to obtain the clean result (\ref{eq:lambda_greater_0.81}), but which adds very little to the understanding of the Ovals Conjecture as such. Finally, Section \ref{sec:proof_estimate_f} fills a gap in the proof of Theorem \ref{theorem:main_results} by establishing an important lemma where lengthy and, again, rather technical work is required.

%==========================================
\section{Preliminaries} \label{sec:Preliminaries}
%==========================================

\subsection{Notation and closure condition}
To prepare the proof of Theorem \ref{theorem:main_results}, this section introduces additional notation and collects some basic facts which have already been established in \cite{Linde:ALowerBound}.

For a given smooth curve $\Gamma$ as in Theorem $\ref{theorem:main_results}$ with an arc length parameter $s$ let $\phi(s)$ be the angle between the tangent on $\Gamma$ in $s$ and some fixed axis, which implies $\phi'(s) = \kappa(s)$. To keep the notation compact we write
$$\phi: \Omega \rightarrow \Omega \quad \mathrm{with}\quad\Omega := \R \slash 2\pi \mathbb Z,$$
considering numbers that differ by an integer multiple of $2\pi$ as identical.
Since the problem is invariant under a reflection of $\Gamma$ in the plane one can assume that $\phi' > 0$ without loss of generality.
For convenience, this and other assumptions made on $\Gamma$ throughout the article are summarized in Table \ref{tab:assumptions}.
The assumption that $\Gamma$ is closed implies
\begin{eqnarray} \label{eq:closure_condition}
&&\int_\Omega \cos \phi(s) \ds =  \int_\Omega \sin \phi(s) \ds = 0,\\
&&\int_\Omega \phi'(s) \ds = 2\pi.\nonumber
\end{eqnarray}
The eigenvalue $\lambda_\Gamma$ is the minimal value - for given $\phi$ - of the energy functional 
\begin{equation}
    E(\phi,\psi) := \frac{\int_\Omega (\psi'(s)^2 + \kappa(s)^2\psi(s)^2)\ds}{\int_\Omega\psi(s)^2\ds}
\end{equation}

It turns out that the curve $\Gamma$ and the lowest eigenfunction of the Schrö\-din\-ger operator $H_\Gamma$ can be characterized in two more ways, which will be called the \emph{$\phi^{-1}$ view} and the \emph{$x$/$y$ coordinates view}. The views are equivalent but in each of them some results are more accessible than in the other. One key idea for proving Theorem \ref{theorem:main_results} is to gather information from the different views and put them into contrast. Each view requires its own notation, which is introduced in the following two subsections.

%------------------------------------------
\subsection{The $\phi^{-1}$ view}
%------------------------------------------

The closure condition on $\Gamma$ can be expressed in a more convenient way in terms of the inverse function $\phi^{-1}: \Omega\rightarrow \Omega$, where the unpleasant integrals over $\cos\phi(s)$ and $\sin\phi(s)$ in (\ref{eq:closure_condition}) turn into the simple demand that the first non-trivial Fourier component of its derivative be zero:
\begin{eqnarray}
&&\int_\Omega (\phi^{-1})'(t) \sin t \dt =  \int_\Omega (\phi^{-1})'(t) \cos t \dt = 0,\\
&&\int_\Omega (\phi^{-1})'(t) \dt = 2\pi.\nonumber
\end{eqnarray}
The function $(\phi^{-1})'$ can therefore be written as a Fourier series
\begin{equation}
    (\phi^{-1})'(t) = 1 + \sum\limits_{n=2}^\infty n a_n \cos nt - n b_n \sin nt,
\end{equation}
such that
\begin{equation}\label{EqPhi}
\phi^{-1}(t) = C + t + g(t) + f(t),
\end{equation}
where $C\in\R$ is some integration constant and
\begin{eqnarray}
    g(t) &:=& \sum\limits_{n=2,4,6,...}^\infty a_{n} \sin nt + b_{n} \cos nt, \label{eq:DefOfG}\\
    f(t) &:=& \sum\limits_{n=3,5,7,...}^\infty a_{n} \sin nt + b_{n} \cos nt.\label{eq:DefOfF}
\end{eqnarray}
Note the periodicity properties
\begin{equation}\label{EqImparityOfF}
    f(t+\pi) = - f(t),\quad g(t+\pi) = g(t) \quad \textmd{for all}\quad t\in\Omega.
\end{equation}
The functions $f$ and $g$ are continuous, since a discontinuity in either of them at some point $t$  would lead to at least one discontinuity of $\phi^{-1}$ in $t$ or in $t+\pi$ due to the periodicity properties (\ref{EqImparityOfF}).

%------------------------------------------
\subsection{Estimates on the total variation of $f$}
%------------------------------------------

\begin{lemma}[Upper bound on total variation of $f$] \label{lem:upper_bound_variation_f}
    The total variation of $f$ satisfies the upper bound
    \begin{equation}\label{eq:upper_bound_variation_f}
        V(f) := \int_\Omega |f'(t)| \dt \le 2\pi.
    \end{equation}
\end{lemma}
\begin{proof}
    To see this, note that $(\phi^{-1})'(t) \ge 0$ since $\phi'(s)>0$. By (\ref{EqPhi}) this means
    \begin{equation*}
        f'(t) + g'(t) \ge -1 \quad \textmd{for all} \, t\in \Omega.
    \end{equation*}
    But, applying (\ref{EqImparityOfF}) to this inequality also implies
    \begin{equation*}
        -f'(t) + g'(t) \ge -1 \quad \textmd{for all} \, t\in \Omega.
    \end{equation*}
    By isolating $f'$ in the last two inequalities and comparing the results one gets
    \begin{equation*}
        |f'(t)| \le 1 + g'(t) \quad \textmd{for all} \, t\in \Omega.
    \end{equation*}
    Integrating over $\Omega$ and keeping in mind the periodicity of $g$ yields (\ref{eq:upper_bound_variation_f}).
\end{proof}

On the other hand, in the main part of the proof we will show that without losing generality the problem can be stated such that the function $f$ satisfies the conditions of the following lemma, yielding a lower bound for the total variation:

\begin{lemma}[Lower bound on total variation of $f$] \label{lem:lower_bounb_variation_f}
    Let $\hat f \in W^{1,1}(\Omega, \R)$, the Sobolev space of absolutely continuous, real-valued functions on $\Omega$, such that the weak derivative $\hat f' \in L^1(\Omega, \R)$. Further, assume that
    \begin{enumerate}
        \item  $\hat f(t + \pi) = -\hat f(t)$ for all $t\in\Omega$ (anti-periodicity),
        \item $\int_\Omega \hat f(t) \sin t\,\dt = \int_\Omega \hat f(t) \cos t\,\dt = 0$ (Fourier condition) and % since \hat f is in W^{1,1} it is also in L^1 and therefore the integrals exist and are finite.
        \item  that $\hat f$ satisfies the following bounds and identities:
        \begin{eqnarray*}
            \hat f(t) &<& 0 \quad \text{ for } t \in [0,\tau_1),\\
            \hat f(\tau_1) &=& \hat f(\tau_2) = 0,\\
            \hat f(t) &>& 0 \quad  \text{ for } t \in [\tau_2,\frac{\pi}{2}),\\
            \hat f(t) &\ge& \delta \quad  \text{ for } t \in [\frac{\pi}{2},\pi]
        \end{eqnarray*}
        for some $\delta > 0$ and for $0 < \tau_1 < \tau_2 < \frac\pi2$.
    \end{enumerate}
    Define
    \[\nu := \max_{t \in [\frac\pi 2,\pi]} \hat f(t) - \delta.\]
    Then the total variation of $\hat f$ is bounded by the inequality
    \begin{equation} \label{eq:lower_bound_variation_f}
        V(\hat f) = \int_\Omega |\hat f'(t)| \dt > 4(\nu+\delta) + \frac{2\sqrt2\delta\sin(\frac{\tau_1 + \tau_2}{2} + \frac\pi4)}{\sin^2(\frac{\tau_2 - \tau_1}{4})}.
    \end{equation}
\end{lemma}
This lemma is proved in Section \ref{sec:proof_estimate_f}.

The r.h.s. of (\ref{eq:lower_bound_variation_f}) can be estimated further using
\[\sin\left(\frac{\tau_1 + \tau_2}{2} + \frac\pi4\right) \ge \sin\left(\frac{\tau_2 - \tau_1}{2} + \frac\pi4\right).\]
This bound can be proved via standard trigonometric identities, showing that the l.h.s. minus the r.h.s. of the inequality can be simplified to the expression
\[2 \cos\left(\frac{\tau_1}2+\frac\pi 4\right)  \sin\frac{\tau_2}2,\]
which is positive for the relevant range of $\tau_1, \tau_2.$
Applying the bound to (\ref{eq:lower_bound_variation_f}) yields the following corollary.
\begin{corollary}[Relaxed lower bound on $V(f)$] \label{cor:relaxed_bound_variation_f}
    Under the assumptions of Lemma \ref{lem:lower_bounb_variation_f}, the estimate on $V(\hat f)$ can be relaxed to:
    \begin{equation} \label{eq:relaxed_bound_variation_f}
        V(\hat f) > 4(\nu + \delta) +\frac{ 2\sqrt 2\delta \sin(\frac{\tau_2 - \tau_1}{2} + \frac\pi4)}{\sin^2(\frac{\tau_2 - \tau_1}{4})}.
    \end{equation}
\end{corollary}

Combining the two bounds on $V(f)$ from Lemma \ref{lem:upper_bound_variation_f} and Lemma \ref{lem:lower_bounb_variation_f}, one can derive estimates on $\delta$ and $\delta + \nu$:

\begin{corollary}[Dual Use Corollary] \label{cor:dual_use}
    Assume that the function $f$ as defined in (\ref{eq:DefOfF}) also satisfies the conditions of Lemma \ref{lem:lower_bounb_variation_f} when $\hat f \equiv f$. Let $\Delta$ be any number such that
    \[\tau_2 - \tau_1 \le \Delta \le \frac\pi2.\]
    Then
    \begin{equation} \label{eq:bound_delta_nu}
        \delta < (\pi -2\nu)G(\Delta)
    \end{equation}
    with $G(\Delta)$ as defined in Theorem \ref{theorem:main_results}.
\end{corollary}
\begin{proof}
    Comparing eq. (\ref{eq:upper_bound_variation_f}) from Lemma \ref{lem:upper_bound_variation_f} with eq. (\ref{eq:relaxed_bound_variation_f}) from Corollary \ref{cor:relaxed_bound_variation_f} directly gives
    \begin{eqnarray*}
        2\pi &>& 4(\nu + \delta) +\frac{ 2\sqrt 2\delta \sin(\frac{\tau_2 - \tau_1}{2} + \pi/4)}{\sin^2(\frac{\tau_2 - \tau_1}{4})}\\
        &\ge& 4(\nu + \delta) +\frac{ 2\sqrt 2\delta\sin(\frac\Delta 2 + \pi/4)}{\sin^2(\frac{\Delta}{4})}\\
    \end{eqnarray*}
    %\begin{equation} \label{eq:def_of_G}
    %    G(\Delta) := \left(\frac{\sin(\frac\Delta 4)}{\sin(\frac{\Delta}{4})+\cos(\frac\Delta 4)}\right)^{2}.
    %\end{equation}
    The last step above is justified because the r.h.s. of the inequality decreases monotonously with $\tau_2-\tau_1$ on the relevant interval, as can be confirmed by a standard differentiation.
    Solving for $\delta$ and applying standard trigonometric identities to simplify the resulting expression yields the statement of Corollary \ref{cor:dual_use}.
\end{proof}

%------------------------------------------
\subsection{Critical points} \label{subsec:critical_points}
%------------------------------------------

Call $s \in \Omega$ a
`critical point' of $\Gamma$ if $\phi(s + \pi) = \phi(s) + \pi$. Obviously, $s+\pi$ also is a critical point then. If
$s$ is a critical point, then $\phi(s)$ will be called a `critical angle'. Due to the periodicity properties of $f$ and $g$, the critical angles are identical to the zeros of $f$.

While open curves may have no critical points at all,
the following lemma holds for the closed curves we are considering:
\begin{lemma} \label{lem:CriticalPoints}
Every smooth closed curve $\Gamma$ has at least six critical points.
\end{lemma}
The proof of this Lemma rests on the anti-periodicity of $f$ and the fact that its first Fourier components must vanish. It is given in detail in \cite{Linde:ALowerBound}. 

It is clear from the definition of a critical point and from Lemma \ref{lem:CriticalPoints}, that every $\Gamma$ has at least three critical points and three critical angles in $[s,s+\pi) \subset \Omega$ for any $s \in \Omega$. 

%------------------------------------------
\subsection{$x$/$y$ coordinates view} \label{sec:xy_view}
%------------------------------------------

For a given curve $\Gamma$ characterized by $\phi(s)$ and some test function $\psi(s) > 0$ in the domain of $H_\Gamma$ one can define the functions
\begin{equation} \label{eq:def_of_x_and_y}
    x(s) := \psi(s) \cos \phi(s), \quad y(s) := \psi(s) \sin \phi(s).
\end{equation}
In particular, since the ground state of $H_\Gamma$ can be chosen to be a positive function, it can also be written in the $x/y$ representation.
Interpreted as Euclidean coordinates, $x$ and $y$ define a closed curve in the plane. In these coordinates, the energy functional assumes the form
\begin{equation}\label{EqLambda}
    E(x,y) = \frac{\int_\Omega \left(x'(s)^2 + y'(s)^2 \right)\ds}{\int_\Omega \left(x(s)^2 + y(s)^2 \right)\ds}.
\end{equation}
In a slight abuse of notation, we will write $E(x,y)$ and $E(\psi)$ interchangeably.
Let
\begin{equation} \label{eq:def_h}
    h_t(s) :=  x(s) \sin t - y(s) \cos t
\end{equation} 
be the orthogonal projections of the curve $(x(s),y(s))$ onto straight lines through the origin. Then 
\begin{equation} \label{eq:def_I}
    I(t) := \frac{\int_\Omega h'_t(s)^2 \ds}{\int_\Omega h_t(s)^2 \ds},
\end{equation}%
will be called the \emph{energy projection} of $\Gamma$ at the angle $t$. It is straight-forward to check that $I(t)$ is $\pi$-periodic, positive, continuous, infinitely often differentiable and bounded.

The following lemma provides the pivotal connection between the $x$/$y$ coordinates view and the $\phi^{-1}$ view:
\begin{lemma} \label{lem:bound_on_I_from_f}
    Let $I(t)$ be the energy projection as defined in (\ref{eq:def_I}) and $f$ the anti-periodic component of $\phi^{-1}$ as defined in (\ref{EqPhi}). Then
    \begin{equation}\label{eq:bound_on_I_from_f}
        I(t) \ge \left(1+\frac{2 |f(t)|}{\pi}\right)^{-2}.
    \end{equation}
\end{lemma}
\begin{proof}
        By putting the definition (\ref{eq:def_of_x_and_y}) of $x$ and $y$ into the definition (\ref{eq:def_h}) of $h_t$ one obtains
        \[h_t(\phi^{-1}(t)) = 0 \quad \text{and}\quad h_t(\phi^{-1}(t+\pi)) = 0  .\]
        These two zeros cut the domain $\Omega$ into an interval of length $\pi+2f(t)$ and one of length $\pi-2f(t)$.
        The energy projection $I(t)$ can be interpreted as the Ritz-Rayleigh ratio for the Laplacian operator on $\Omega$ with Dirichlet conditions and $h_t(s)$ as one specific test function. Thus $I(t)$ must be larger or equal the lowest eigenvalue of the Laplacian, which yields (\ref{eq:bound_on_I_from_f}).
\end{proof}

Since a critical angle is a zero of $f$, Lemma \ref{lem:bound_on_I_from_f} states that $I(t) \ge 1$ if $t$ is a critical angle.
If the critical angles of $\Gamma$ are distributed somewhat evenly then this statement can be extended to an estimate on $\lambda_\Gamma$:

\begin{lemma} \label{lem:EvenlyDistributedCriticalAngles}
Let $\Gamma$ be as in Theorem \ref{theorem:main_results} and assume additionally that every interval $[\phi,\phi+\frac\pi 2)
\subset \Omega$ contains at least one critical angle of $\Gamma$. Then $\lambda_\Gamma \ge 1$.
\end{lemma}
The lemma was proved as Theorem 2.2 in \cite{Linde:ALowerBound}. It is also not difficult to derive it from the Three Angles Lemma (Lemma \ref{lem:three_angles}) that will be presented in Section \ref{subsec:two_ways} below. %Assume 0 < alpha < beta < gamma < pi to be critical angles with beta-alpha<pi/2 and gamma-beta<pi/2. Then one can chack that a,b,c are positive. So are the V_alpha times X terms and the I(alpha),... are greater than one.

\begin{lemma} \label{lem:facts_about_I}
    Either 
    
    $I(t)$ is constant and then $\lambda_\Gamma\ge 1$, or 
    
    $I(t)$ has exactly one pair of maxima at a distance of $\pi$ from each other and exactly one pair of minima also at a distance of $\pi$ from each other and it is strictly monotonous in between those extrema.
\end{lemma}

\begin{proof}
    If $I(t) = I$ is constant then $I(t) \ge 1$ for all $t\in\Omega$ since we already know that this inequality holds if $t$ is any of the at least six critical angles of $\Gamma$. In that case we also have
    \[I = I(0) = \frac{\int_\Omega y'(s)^2\ds}{\int_\Omega y(s)^2\ds} = I(\frac{\pi}{2}) = \frac{\int_\Omega x'(s)^2\ds}{\int_\Omega x(s)^2\ds}\]
    and therefore by (\ref{EqLambda}) we get  $\lambda_\Gamma = I \ge 1$.
    
    So assume now that $I(t)$ is not constant. Since $I(t)$ is $\pi$-periodic, the following analysis can be restricted to the case $t \in (-\frac\pi2,\frac\pi2)$ where $\cos t>0$ plus the edge case $t = \frac\pi2$.
    Put (\ref{eq:def_h}) into (\ref{eq:def_I}), simplify the fraction with $\cos^2 t$ and write shorthand $\tau(t) = \tan t$ to obtain
    \begin{equation*}
         I(t) = \frac{(\int_\Omega x'^2\ds) \tau^2 -(\int_\Omega 2x'y' \ds) \tau +(\int_\Omega y'^2 \ds)}{(\int_\Omega x^2\ds) \tau^2 -(\int_\Omega 2xy \ds) \tau +(\int_\Omega y^2 \ds)}.
    \end{equation*}
    By the quotient rule the derivative of $I(t)$ can be written as
    \[I'(t) = \frac{\tau' P_3(\tau)}{P_4(\tau)}\]
    where $P_3(\tau)$ is a polynomial in $\tau$ of up to third degree and $P_4(\tau)$ is a polynomial in $\tau$ of exactly fourth degree.

    Being $\pi$-periodic but not constant, $I(t)$ must have a positive even number of extrema on $(-\frac\pi 2,\frac\pi 2]$, alternating between maxima and minima.
    The case of two extrema is claimed in the lemma, so it only remains to show that the number of extrema cannot be four or higher.

    For $I'(t)$ to have at least four zeros on $(-\frac\pi 2,\frac\pi 2]$, the polynomial $P_3$ must be of third degree to have three zeros on $(-\frac\pi 2,\frac\pi 2)$ and there must be a fourth zero at $t = \frac\pi 2$. Yet this fourth zero cannot exist: Being of third degree, $P_3(\tau)$ behaves like some (non-zero) constant times $\tau^3$ as $\tau \rightarrow \infty$ (\ite $t \rightarrow \frac\pi 2$). It is multiplied by $\tau'$ which increases like $\tau^2$ as $\tau \rightarrow \infty$. Their product then behaves like $\tau^5$. Their quotient with $P_4$, \ite the derivative $I'(t)$, can thus not converge to zero as $\tau\rightarrow\infty$. The continuity of $I'(t)$ then ensures that $I'(\frac\pi 2) \neq 0$.
\end{proof}

%------------------------------------------
\subsection{Short-hand notation}
%------------------------------------------

To streamline the algebra later on, it will be convenient to write the integrals above as scalar products in $\R^3$. To that end, define
\begin{equation}
    X := \int_\Omega
    \begin{pmatrix}
        x(s)^2\\
        -2x(s)y(s)\\
        y(s)^2
    \end{pmatrix}
    \ds,\, 
\end{equation}
\begin{equation}
    \hat X := \int_\Omega
    \begin{pmatrix}
        x'(s)^2\\
        -2x'(s)y'(s)\\
        y'(s)^2
    \end{pmatrix}
    \ds,
\end{equation}
\begin{equation} \label{eq:def_Vt_and_N}
    V_t= 
    \begin{pmatrix}
        \sin^2 t \\
        \sin t\cos t\\
        \cos^2 t
    \end{pmatrix}
    ,\text{ and }
    N = 
    \begin{pmatrix}
        1 \\
        0 \\
        1
    \end{pmatrix}.
\end{equation}
Then $E(x,y)$ and $I(t)$ assume the simple forms 
\begin{equation} \label{eq:NX_and_VX}
    E(x,y) = \frac{N \cdot \hat X}{N \cdot X} \quad \text{and} \quad I(t) = \frac{V_t \cdot \hat X}{V_t \cdot X}
\end{equation}

This completes all preliminaries required for the proof of Theorem \ref{theorem:main_results}.

%==========================================
\section{Proof of Theorem \ref{theorem:main_results}}  \label{sec:Proof}
%==========================================
%------------------------------------------
\subsection{Overview and main ideas}
%------------------------------------------
Before going into the details of the proof of Theorem \ref{theorem:main_results}, the following is a brief outline of the main ideas: 

The lowest eigenvalue $\lambda_\Gamma$ can be interpreted as some weighted average of energy projections $I(t)$ for different angles $t$ in the $x$/$y$ coordinates view. The function $I(t)$ has two maxima and two minima on $\Omega$ and it is monotonous in between. Since $I(t) \ge 1$ at critical angles, there are two intervals in $\Omega$ where $I(t)$ is larger than one and two intervals where it is smaller, unless $\lambda_\Gamma \ge 1$. Loosely speaking, if the interval where $I(t) \ge 1$ is too large (\ite, close to $\frac{\pi}{2}$), then the energy projections over this range of different angles already specify $\lambda_\Gamma$ so tightly that it can't be much smaller than $1$. On the other hand, if the interval is too small (\ite, close to zero) then the critical angles of $\Gamma$ - which are contained in that interval - are very close together. But the critical angles are the zeros of $f$ in the $\phi^{-1}$ view and if they are contained in a small interval, it restricts the ways how $f$ can be `balanced' such that its first Fourier components vanish. This will lead to a bound on $f$, which in turn leads to yet another bound on $\lambda_\Gamma$. The `worst case' is somewhere in between, where the interval with $I(t) \ge 1$ is neither too small nor too large, and it will serve as the basis for the global estimate on $\lambda_\Gamma$.

%------------------------------------------
\subsection{Two ways of estimating $\lambda_\Gamma$} \label{subsec:two_ways}
%------------------------------------------

In order to realize this road map for the proof, two important lemmas are required. The first one makes use of the $x$/$y$ coordinates view of the problem and it states how $\lambda_\Gamma$ can be expressed in terms of three different energy projections:
\begin{lemma}[Three Angles Lemma] \label{lem:three_angles}
    Let $\alpha, \beta, \gamma \in \Omega$ be angles such that no two of them are congruent modulo $\pi$. In other words, for any pair $\theta_1, \theta_2$ from $\{\alpha, \beta, \gamma\}$, one has
    \[\theta_1 \not\equiv \theta_2 \quad (\text{mod }\pi).\]
    Then 
    \begin{equation} \label{eq:three_angles_lemma}
        E(x,y) = \frac{aI(\alpha) X \cdot V_\alpha + b I(\beta)X \cdot V_\beta + cI(\gamma) X \cdot V_\gamma}{a X \cdot V_\alpha + bX \cdot V_\beta + cX \cdot V_\gamma}
    \end{equation}
    where
    \begin{eqnarray*}
        a &=& \frac{\cos(\beta - \gamma)}{\sin(\alpha - \beta) \sin(\alpha-\gamma)}\\
        b &=& \frac{\cos(\alpha - \gamma)}{\sin(\beta - \alpha) \sin(\beta - \gamma)}\\
        c &=& \frac{\cos(\alpha - \beta)}{\sin(\gamma - \alpha) \sin(\gamma - \beta)}
    \end{eqnarray*}
\end{lemma}
The condition on $\alpha, \beta, \gamma$ being different modulo $\pi$ ensures that the denominators in the expressions for $a$, $b$ and $c$ are nonzero so that the three parameters are properly defined. % This can be seen be dividing the brackets by the respective cos. Then a bracket can only vanish if tan alpha = tan beta (or analogous for other pairs of angles), which is excluded by the condition. Only exception: Where tan blows up. But those poles are also pi apart and therefore it is excluded that the bracket vanishes.
\begin{proof}
    Putting the expressions for $a$, $b$ and $c$ from Lemma \ref{lem:three_angles} and the definition of $V_t$ from (\ref{eq:def_Vt_and_N}) into the l.h.s. of (\ref{eq:N_decomposed}) below shows after a tedious but straightforward exercise in algebra that
    \begin{equation}\label{eq:N_decomposed}
        a V_\alpha + bV_\beta + cV_\gamma =
        \begin{pmatrix}
            1\\
            0\\
            1
        \end{pmatrix}
        = N.
    \end{equation}
    Put this decomposition of $N$ into the expression for $E(x,y)$ in (\ref{eq:NX_and_VX}). Then replace the three $V_t\hat X$ terms by the respective $I(t)V_t X$, using the identity for $I(t)$ also from (\ref{eq:NX_and_VX}). The result is (\ref{eq:three_angles_lemma}).
\end{proof}

\begin{corollary} \label{cor:two_angles_corollary}
    For any $\alpha\in\Omega$,
    \begin{equation}
       E(x,y) \ge \min\left(I(\alpha), I(\alpha+\frac\pi 2)\right).
    \end{equation}
\end{corollary}
\begin{corollary} \label{cor:two_angles_equal_I_corollary}
    If $I(\alpha) = I(\alpha+\frac\pi 2)$ for any $\alpha\in\Omega$, then $E(x,y) = I(\alpha)$.
\end{corollary}
\begin{proof}
     Set $\beta = \alpha+\frac{\pi}{2}$ and let $\gamma \in \Omega$ be arbitrary but different from $\alpha$ and $\beta$ modulo $\pi$. Then the parameters in Lemma \ref{lem:three_angles} become
     \[a = b = 1 \text{ and } c = 0.\]
    Thus, according to Lemma \ref{lem:three_angles}, 
    \begin{equation*}
        E(x,y) = \frac{I(\alpha) X \cdot V_\alpha + I(\alpha+\frac{\pi}{2})X \cdot V_\beta}{X \cdot V_\alpha + X \cdot V_\beta }
    \end{equation*}
    which is a weighted average with positive coefficients of $I(\alpha)$ and $I(\alpha+\frac{\pi}{2})$ and proves both Corollary \ref{cor:two_angles_corollary} and Corollary \ref{cor:two_angles_equal_I_corollary}.
\end{proof}

Lemma \ref{lem:three_angles} opens up one way to find estimates for $\lambda_\Gamma$. The following lemma provides a second, complimentary approach. Both will be combined in the proof of Theorem \ref{theorem:main_results}.

\begin{lemma} \label{lem:estimate_lambda_interval}
    Let $\Gamma$ be as in Theorem \ref{theorem:main_results} and let $\{t_i\}_{i=1...n} \subset \Omega$ be a set of numbers such that $[t,t+\frac\pi 2] \cap \{t_i\} \neq \emptyset$ for all $t\in\Omega$. Assume that $|f(t_i)| \le \delta$ for all $i$ and for some $\delta > 0$.
    Then
    \begin{equation} \label{eq:EstimateLambdaVsAlpha}
        \lambda_\Gamma \ge \left(1+2\delta/\pi\right)^{-2}.
    \end{equation}
\end{lemma}

Lemma \ref{lem:estimate_lambda_interval} is proved in \cite{Linde:ALowerBound} and we refer the reader to that publication for the details of the proof. Intuitively speaking, the condition $|f(t_i)| \le \delta$ in the lemma is a measure for how far the curvature of $\Gamma$ is away from being $\pi$-periodic, since $f(t)$ is the only term on the right hand side of (\ref{EqPhi}) which does not give a $\pi$-periodic contribution to $\kappa$. For $\pi$-periodic curvatures the conjectured bound $\lambda_\Gamma \ge 1$ can easily be established and therefore controlling the deviation from periodicity allows to estimate $\lambda_\Gamma$ in the general case.

%------------------------------------------
\subsection{Proof of Theorem \ref{theorem:main_results}}
%------------------------------------------

Let $\Gamma$ be the curve from Theorem \ref{theorem:main_results} and $\psi>0$ the lowest eigenfunction of the Schrödinger operator $H_\Gamma$. Then $\Gamma$ determines the function $\phi(s)$, and the pair $\{\phi(s), \psi(s)\}$ defines the functions $x(s)$ and $y(s)$ as introduced in Section \ref{sec:xy_view}. The latter functions, in turn, determine the energy projection $I(t): \Omega \rightarrow \R_+$.

As mentioned in Section \ref{subsec:critical_points}, one has $I(t)\ge 1$ if $t$ is one of the at least six critical angles of $\Gamma$. On the other hand, one can assume that there is some $t\in\Omega$ where $I(t) < 1$. Otherwise, by Corollary \ref{cor:two_angles_corollary} one has $\lambda_\Gamma = E_\Gamma(\phi,\psi) \ge 1$ and the claim from Theorem \ref{theorem:main_results} is trivially correct.

Thus $I(t)$ attends a value greater than one at the two maxima which are guaranteed by Lemma \ref{lem:facts_about_I}, and a value below one at the two minima. 

It follows from the periodicity and monotonicity properties on $I(t)$ and a standard application of the intermediate value theorem that there is exactly one $t_\lambda \in [0, \frac{\pi}{2})$ such that
\[I(t_\lambda) = I(t_\lambda + \frac\pi2) = I(t_\lambda + \pi) = I(t_\lambda + \frac{3\pi}{2}).\]
Due to Corollary \ref{cor:two_angles_equal_I_corollary} this implies $\lambda_\Gamma = I(t_\lambda)$. Of course, we assume that $\lambda_\Gamma < 1$, for otherwise there is nothing to prove.

Without loss of generality, assume that $\Gamma$ is rotated such that $t_\lambda = 0$ and $I(t) > \lambda_\Gamma$ for $t \downarrow 0$. Note that for this rotation of $\Gamma$ there are two possible choices which differ by an angle of $\pi$, due to the $\pi$-periodicity of $I(t)$. We will decide for one of the two options later on.

It is a further consequence from the intermediate value theorem that there must be exactly four points $\iota_1, \dots, \iota_4$ where $I(\iota_n) = 1$ for $n = 1\dots 4$.

The indices of the $\iota_n$ can be chosen such that $0 < \iota_1 < \iota_2 < \frac{\pi}{2}$. The periodicity of $I(t)$ and a further proper choice of indices then imply $\iota_3 = \iota_1 + \pi$ and $\iota_4 = \iota_2 + \pi$.

We have seen before that $I(t)\ge 1$ at critical angles and therefore the intervals $[\iota_1,\iota_2]$ and $[\iota_3,\iota_4]$ contain all the critical angles of $\Gamma$, of which there are at least three pairs.

In the following, the smallest and the largest critical angle in $[\iota_1,\iota_2]$ will become relevant and they will be denoted by $\tau_1$ for the smallest and $\tau_2$ for the largest.
Note that a smallest and a largest always exist even in the case of infinitely many critical angles, since they are the zeros of the continuous function $f(t)$ and thus form a closed set.
 
For the reader's convenience, Figure \ref{fig:plot_I} visualizes the energy projection and the special points that have been defined and Table \ref{tab:assumptions} provides a summary of the choices that have been made w.l.o.g.
    
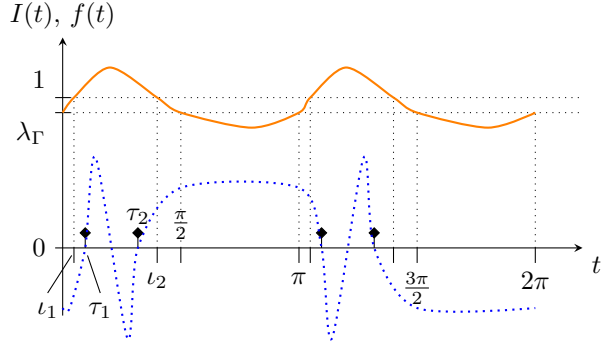
\begin{figure} 
    \centering
    \begin{tikzpicture}[>=stealth, xscale = 1, yscale = 2]

\pgfmathsetmacro{\iotaone}{0.15}
\pgfmathsetmacro{\iotatwo}{0.4*pi}
\pgfmathsetmacro{\tauone}{0.3}
\pgfmathsetmacro{\tautwo}{1.0}
\pgfmathsetmacro{\lamda}{0.9}

% draw the axes ------------------------------------------------

\draw[->] (  0,-0.45) -- (0, 1.4) node[above] {$I(t)$, $f(t)$};
\draw[->] ( -0. 1,0) node[left] {$0$} -- ({2.2*pi},0) node[below right] {$t$};
\draw[dotted] (0.1,  1) -- ({2.2*pi}  ,1);
\draw[dotted] (0.1,  0.9) -- ({2.2*pi}  ,\lamda);

% vertical dotted lines at the iotas
\draw[dotted] (\iotaone, 0) -- (\iotaone, 1.0);
\draw[dotted] (\iotatwo, 0) -- (\iotatwo, 1.0);

\draw[dotted] ({\iotaone + pi}, 0) -- ({\iotaone + pi}, 1.0);
\draw[dotted] ({\iotatwo + pi}, 0) -- ({\iotatwo + pi}, 1.0);

% vertical dotted lines at the multiples of pi
\draw[dotted] ({0.5*pi}, 0) -- ({0.5*pi}, \lamda);
\draw[dotted] ({pi},     0) -- ({pi},     \lamda);
\draw[dotted] ({1.5*pi}, 0) -- ({1.5*pi}, \lamda);
\draw[dotted] ({2*pi},   0) -- ({2*pi},   \lamda);

% y axis labels

\draw (0.1,  1.0) -- (-0.1,  1.0) node[above left] {$1$};
\draw (0.1,  0.9) -- (-0.1,  0.9) node[below left] {$\lambda_\Gamma$};

% x axis labels ------------------------------------------------

% tau's (critical points)
\draw (\tauone, -0.0) -- (\tauone,  0.1);
\draw (\tauone+0.03, -0.05) -- (\tauone+0.2, -0.3) node[below] {$\tau_1$};
\draw (\tauone,  0.1) node[draw, diamond, fill=black, inner sep=0pt, minimum size=4pt] {};
\draw (\tautwo,  -0.0) -- (\tautwo, 0.1) node[above] {$\tau_2$};
\draw (\tautwo,  0.1) node[draw, diamond, fill=black, inner sep=0pt, minimum size=4pt] {};
\draw ({\tauone+pi}, -0.0) -- ({\tauone+pi},  0.1) node[above] {};
\draw ({\tauone+pi},  0.1) node[draw, diamond, fill=black, inner sep=0pt, minimum size=4pt] {};
\draw ({\tautwo+pi},  0.1) -- ({\tautwo+pi}, -0.0) node[below] {};
\draw ({\tautwo+pi},  0.1) node[draw, diamond, fill=black, inner sep=0pt, minimum size=4pt] {};

% iota's
\draw (\iotaone, 0.0) -- ( \iotaone, -0.1);
\draw ({\iotaone-0.03}, -0.05) -- (-0.2, -0.3) node[below] {$\iota_1$};
\draw (\iotatwo, 0.0) -- ( \iotatwo, -0.1) node[below] {$\iota_2$};
\draw ({\iotaone +  pi}, 0.0) -- ({\iotaone +  pi}, -0.1) node[above] {};
\draw ({\iotatwo +  pi}, 0.0) -- ({\iotatwo +  pi}, -0.1) node[above] {};

% Multiples of pi
\draw ({pi/2},-0.1) -- ({pi/2},   0.0) node[above] {$\frac{\pi}{2}$};
\draw ({pi}  ,   0) -- ({pi}  ,  -0.1) node[below] {$\pi$};
\draw ({3*pi/2}, 0) -- ({3*pi/2},-0.1) node[below] {$\frac{3\pi}{2}$};
\draw ({2*pi},   0) -- ({2*pi}  ,-0.1) node[below] {$2\pi$};

% Function I(t) ----------------------------------------------
\draw[\orange,thick]
  plot[smooth] coordinates {
    ( { 0.0 * pi},      \lamda)
    (    \iotaone,      1.0)
    ( { 0.2 * pi},      1.2)
    (    \iotatwo,      1.0)
    ( { 0.5 * pi},      \lamda)
    ( { 0.8 * pi},      0.8)
    ( { 1.0 * pi},      \lamda)
    ( {\iotaone + pi},  1.0)
    ( { 1.2 * pi},      1.2)
    ( {\iotatwo + pi},  1.0)
    ( { 1.5 * pi},      \lamda)
    ( { 1.8 * pi},      0.8)
    ( { 2.0 * pi},      \lamda)
  };

% Function f(t) ----------------------------------------------
\draw[\blue, thick, dotted]
  plot[smooth] coordinates {
    ( { 0.0 * pi},     -0.4)
    ( { 0.03 * pi},     -0.4)
    (     \tauone,      0.0)
    ( {(0.8*\tauone+0.2*\tautwo)},  0.6)
    ( {(0.2*\tauone+0.8*\tautwo)}, -0.6)
    (    \tautwo,      0.0)
    ( {0.5 * pi},      0.4)
    ( {1.0 * pi},      0.4)
    ( {\tauone+pi},    0.0)
    ( {(pi+0.8*\tauone+0.2*\tautwo)}, -0.6)
    ( {(pi+0.2*\tauone+0.8*\tautwo)},  0.6)
    ( {pi+\tautwo},      0.0)
    ( {1.5 * pi},      -0.4)
    ( {2.0 * pi},      -0.4)
  };

\end{tikzpicture}
    \caption{By assumption, the $\pi$-periodic energy projection $I(t)$ \ifcolor (orange line) \else (solid line) \fi is equal to one at $\iota_1 > 0$ and $\iota_2 < \frac{\pi}{2}$ and it is larger than one in between. Four critical angles of $\Gamma$ are marked with diamonds in the chart. They occur in pairs which are separated by a distance of $\pi$, respectively. All critical angles below $\pi$ are contained in $[\iota_1,\iota_2]$ and $\tau_1$ and $\tau_2$ are the smallest and the largest of those, respectively. The critical angles are zeros of $f(t)$ \ifcolor (blue dotted line) \else (dotted line)\fi.}
    \label{fig:plot_I}
\end{figure}

\subsubsection*{Estimate 1 (based on the formula from \cite{Linde:ALowerBound})}

% Problem: Schätzung an lambda kommt von der Möglichkeit, drei verteilte
% Punkte mit großem I zu finde. Verhindert werden kann diese Möglichkeit
% dadurch, dass es IRGANDWO ein Interval der Länge pi/2 gibt, wo f 
% klein ist

% Zu zeigen: In jedem Intervall der Länge pi/2 finde ich einen Punkt 
% Wo I > ... . Das zeige ich, indem ich zeige, dass es in jedem Intervall
% der Länge pi/2 einen Punkt gibt, wo |f| < delta

% Option 1: Die einfache Form von I mit 2 Maxima und zwei Minima nutzen:
%   I hat kein Maximum in [pi/2, pi]. Also ist es monoton oder hat ein Minimum und ist drum herum monoton. f könnte gleich delta sein bei pi/2 aber größer bei pi. I an den beiden Randpunkten abzuschätzen reicht also nicht
% Option 2: Das existierende Lemma nutzen mit allen denkbaren Verschiebungen des Intervalls
% Option 3: Die kurve so drehen, dass f= delta sowohl bei pi/2 alsauch bei pi
%    ==> geht nicht unbedingt, weil f wild schwanken kann
% Option 4: Die Kurve so drehen, dass I = lambda bei pi/2 alsauch bei pi. Dann
% muss |f| bei pi/2 und pi eine gewisse Mindestgröße delta haben. Und wegen der
% Monotonieeigenschaften und dem Minimum von I muss |f| dann sogar noch 
% größer sein im Inneren des Intervalls. 

By the assumptions made so far, $I(t) > \lambda_\Gamma$ on $(0,\frac\pi2)$ and $I(t) \le \lambda_\Gamma$ on $[\frac\pi2, \pi]$. 
Corollary \ref{cor:two_angles_corollary} then implies that $E(x,y) \ge I(t)$ for any $t\in[\frac\pi2,\pi]$. But Lemma \ref{lem:bound_on_I_from_f} states that
\begin{equation} \label{eq:bound_lambda_estimate_0}
    I(t) \ge \left(1 + \frac{2|f(t)|}{\pi}\right)^{-2}  \quad\text{for all } t\in\Omega.
\end{equation}
Therefore 
\begin{equation} \label{eq:bound_lambda_estimate_1}
    \lambda_\Gamma = E(x,y) \ge I(t) \ge \left(1 + \frac{2|f(t)|}{\pi}\right)^{-2}  
\end{equation}
for any $t\in[\frac\pi2,\pi]$.

Since $f$ is continuous and we assume that $\lambda_\Gamma < 1$, this estimate implies that the function $f$ cannot come arbitrarily close to zero or even change signs on $[\frac\pi2, \pi]$. It is also anti-periodic with period length $\pi$, so it is either strictly positive on $[\frac\pi2, \pi]$ and strictly negative on $[\frac{3\pi}2, 2\pi]$, or vice versa.
The rotation of the curve that we have chosen w.l.o.g. so far was only determined up to an angle of $\pi$ since $I(t)$ is $\pi$-periodic. Thus we can use the remaining degree of freedom to assume w.l.o.g. that $\Gamma$ is rotated such that $f(t)>0$ on $[\frac\pi2,\pi]$.

Then $f$ satisfies the assumptions on $\hat f$ in Lemma \ref{lem:lower_bounb_variation_f} with any appropriate choice of $\delta > 0$.
By Corollary \ref{cor:dual_use} the value of $\delta$ is then bounded by
\begin{equation} \label{eq:bound_delta_nu_2}
    \delta < (\pi -2\nu)G(\Delta)
\end{equation}
where we choose to set
\begin{equation} \label{eq:def_of_Delta}
    \Delta := \iota_2-\iota_1    
\end{equation}
and
\[\nu = \max_{t\in[\frac\pi2,\pi]} f(t)-\delta.\]
Consequently, there must be some $t\in[\frac\pi2,\pi]$ where $f(t)$ is not larger than the r.h.s. of (\ref{eq:bound_delta_nu_2}), for otherwise $\delta$ could be chosen larger.

For this particular choice of $t$, inequality (\ref{eq:bound_lambda_estimate_1}) then becomes
\begin{align}
    \lambda_\Gamma &\ge \left(1 + 2\left(1 -\frac{2\nu}\pi\right)G(\Delta)\right)^{-2}\nonumber\\
    &= \frac{1}{\left(1 + 2\tilde\nu G(\Delta)\right)^{2}}\label{eq:lambda_estimate_1}\\
    &= B_1(\tilde\nu, \Delta) \nonumber
\end{align}
introducing the short-hand notation
\[\tilde\nu = \left(1 -\frac{2\nu}\pi\right).\]

\subsubsection*{Estimate 2 (based on the Three Angles Lemma)}

To obtain a second, complementary bound on $\lambda_\Gamma$, apply the Three Angles Lemma with 
\begin{align*}
    \alpha &= \iota_1,\\
    \beta &= \iota_2,\\
    \gamma &= (\iota_1 + \iota_2)/2 + \frac\pi2 \in (\frac\pi2,\pi)
\end{align*}
Since then $I(\alpha) = I(\beta) = 1$, Lemma \ref{lem:three_angles} yields
\begin{align} \label{eq:lambe_estimate_c_I}
    \lambda_\Gamma &= 1 + c (I(\gamma)-1) \frac{X\cdot V_\gamma}{aX\cdot V_\alpha + bX\cdot V_\beta + cX\cdot V_\gamma} \nonumber\\
    &=1 - c (1-I(\gamma)) \frac{X\cdot V_\gamma}{X \cdot N},
\end{align}
where the last step uses once again the equation (\ref{eq:N_decomposed}).
Note that the constant $c$ as defined in Lemma \ref{lem:three_angles} is positive. The term $(1-I(\gamma))$ is also positive because
\[I(\gamma) \le \lambda_\Gamma < 1 \text{ for } \gamma \in (\frac\pi2,\pi).\]
Finally, the fraction can be estimated by
\begin{align*}
    0 < \frac{X\cdot V_\gamma}{X \cdot N} &= \frac{\int_\Omega (x(s) \sin\gamma - y(s)\cos\gamma)^2 \ds}{\int_\Omega(x(s)^2 + y(s)^2) \ds}\\
     &= 1- \frac{\int_\Omega (x(s) \cos\gamma + y(s)\sin\gamma)^2 \ds}{\int_\Omega(x(s)^2 + y(s)^2) \ds}\\
     &< 1
\end{align*}
Consequently, equation (\ref{eq:lambe_estimate_c_I}) can be transformed to the bound
\begin{equation} \label{eq:estimate_lambda_c_I_gamma}
    \lambda_\Gamma > 1 - c(1-I(\gamma)).
\end{equation}
Since $\gamma \in (\frac\pi2,\pi)$ one has
\[f(\gamma) \le \delta + \nu\]
and therefore by Lemma \ref{lem:bound_on_I_from_f}
\[I(\gamma) \ge \left(1+\frac{2\delta + 2\nu}{\pi}\right)^{-2}.\]
Once again, the estimate on $\delta$ from Corollary \ref{cor:dual_use} can be applied to obtain
\begin{align*} 
    I(\gamma) &\ge \left(1+ \frac{2\nu}\pi+ \left(2 -\frac{4\nu}\pi\right)G(\Delta)\right)^{-2}\\
\end{align*}
where $\Delta$ is still defined as in (\ref{eq:def_of_Delta}).
The constant $c$ is defined by Lemma \ref{lem:three_angles}. Using the definitions of $\gamma$ and $\Delta$ and applying standard trigonometric identities, $c$ can be shown to be
\begin{align*}
    c &=  \frac{\cos(\iota_1-\iota_2)}{\sin(\gamma - \iota_1) \sin(\gamma - \iota_2)}\\
    %&=  \frac{\cos(\iota_1-\iota_2)}{\sin((-\iota_1 + \iota_2)/2 + \frac\pi2) \sin((\iota_1 - \iota_2)/2 + \frac\pi2)}\\
    %&=  \frac{\cos(\iota_1-\iota_2)}{\cos((\iota_1 - \iota_2)/2) \cos((\iota_1 - \iota_2)/2)}\\
    %&=  \frac{\cos(\iota_1-\iota_2)}{\cos^2((\iota_1 - \iota_2)/2)}\\
    %&=  \frac{\cos\Delta}{\cos^2(\Delta/2)}\\
    &= 2 - \sec^2 \frac\Delta 2
\end{align*}
Putting this representation of $c$ and the estimate for $I(\gamma)$ into (\ref{eq:estimate_lambda_c_I_gamma}) yields
\begin{align}
    \lambda_\Gamma &> 1 -  \left(2 - \sec^2 \frac\Delta 2\right)\left( 1-\left(1+ \frac{2\nu}\pi+ \left(2 - \frac{4\nu}\pi\right) G(\Delta)\right)^{-2}\right)\nonumber\\
    % &= 1 -  \left(2 - \sec^2 \frac\Delta 2\right)\left( 1-\left(1+ \frac{2\nu}\pi+ 2G(\Delta) - \frac{4\nu}\pi G(\Delta)\right)^{-2}\right)\\
    %&= 1 -  \left(2 - \sec^2 \frac\Delta 2\right)\left( 1-\left(2+(1-\frac{2\nu}{\pi})(2G(\Delta)-1)\right)^{-2}\right)\\
    &= 1 -  \left(2 - \sec^2 \frac\Delta 2\right)\left( 1-\frac{1}{\left(2+\tilde\nu (2G(\Delta)-1)\right)^{2}}\right) \label{eq:lambda_estimate_2}\\
     &= B_2(\tilde\nu, \Delta) \nonumber
\end{align}

\subsubsection*{Putting both estimates together}

In summary, we have found two complementary estimates for $\lambda_\Gamma$ which both only depend on the distance $\Delta$ of points where the energy projection is equal to one, and on the peak value of $f$ on $[\frac\pi2,\pi]$ represented (indirectly) by $\tilde\nu$.
By assumptions on $\iota_1$ and $\iota_2$, it is clear that $0<\Delta<\frac\pi2$.
Since $f$ is anti-periodic and its total variation is bounded by $2\pi$, its range is contained in $[-\frac\pi2,\frac\pi2]$. Thus, by definition of $\nu$, one knows that
\[0 \le \nu < \max_{t\in[\frac{\pi}{2},\pi]} f(t) \le \frac\pi2\]
and therefore $0 \le \tilde\nu \le 1$.
Combining (\ref{eq:lambda_estimate_1}) and (\ref{eq:lambda_estimate_2}) we arrive at the bound (\ref{eq:main_result}), thus proving Theorem \ref{theorem:main_results}.

%==========================================
\section{Proof of Theorem \ref{theorem:main_result2}}  \label{sec:Proof2}
%==========================================

\subsection{Overview and main ideas}

The strategy for proving Theorem \ref{theorem:main_result2} is as follows: We will characterize the level curve $\mathcal{L}$ defined by $B_1(\tilde\nu,\Delta) = 0.81$ in $[0,1]\times(0,\frac\pi 2)$. From the monotonicity properties of $B_1$ we can then conclude that $B_1(\tilde\nu,\Delta)$ is below $0.81$ only on one side of $\mathcal{L}$, \ite in the hatched area of Figure \ref{fig:plot_B1_B2}. So it remains to prove that $B_2$ is not smaller than $0.81$ in that area. Since $B_2$ increases monotonously with $\tilde\nu$, its values in the hatched area are bounded by the values it attains on $\mathcal{L}$ itself. So it only remains to prove that \[B_2(\tilde\nu,\Delta)|_{\mathcal{L}} \ge 0.81\]
Doing so is lengthy and technical, but not difficult.

\subsection{Characterizing $\mathcal{L}$}

To realize this outline of the proof, we start by summarizing certain monotonicity properties which can all be established via standard arguments:
For $\Delta\in(0,\frac\pi 2)$ and $\tilde\nu\in[0,1]$
\begin{itemize}
    \item $G(\Delta)$ increases strictly monotonously and since $G(0)=0$ this implies that $G(\Delta)$ is positive,
    \item $B_1(\tilde\nu, \Delta)$ decreases strictly monotonously with $\tilde\nu$ and with $\Delta$, and
    \item $B_2(\tilde\nu, \Delta)$ increases strictly monotonously with $\tilde\nu$.
\end{itemize}
Consider the level set
\[\mathcal{L} := \{(\tilde\nu,\Delta)|B_1(\tilde\nu,\Delta) = 0.81\}.\]
A point $(\tilde\nu,\Delta)$ being in $\mathcal{L}$ implies
\[(1+2\tilde\nu G(\Delta))^2 = \left(\frac{10}{9}\right)^2.\]
When taking the square root of both sides of the equation, the remaining terms are both positive since $\tilde\nu$ and $G(\Delta)$ are non-negative. Thus, solving for $\tilde\nu$ shows that all points in $\mathcal{L}$ satisfy the equation
\begin{equation} \label{eq:condition_cal_L}
    \tilde\nu =\frac{1}{18G(\Delta)}.
\end{equation}
%\begin{align*}
%    \left(\frac{9}{10}\right)^2 &= \frac{1}{(1+2\tilde\nu G(\Delta))^2}\\
%    \left(\frac{10}{9}\right)^2 &= (1+2\tilde\nu G(\Delta))^2\\
%    \pm \frac{10}{9} &= 1+2\tilde\nu G(\Delta)\\
%    \tilde\nu &= \frac{1}{18G(\Delta)}
%\end{align*}
%By standard arguments it can be shown that $G(\Delta)$ is convex and increases strictly monotonously on $(0,\frac\pi2)$. Direct computation using
%\[\sin\frac\pi 8 = \frac{\sqrt{2-\sqrt 2}}{2},\quad\cos\frac\pi 8 = \frac{\sqrt{2+\sqrt 2}}{2}\]
%shows
%\begin{align*}
%    G(0) &= 0,\\
%    G(\frac\pi 2) &= \frac 32-\sqrt 2 =: G_{\text{max}} \approx 0.086
%\end{align*}
It prescribes exactly one positive value $\tilde\nu$ for every $\Delta\in(0,\frac\pi 2)$, but not every such $\tilde\nu$ is smaller than one and thus in the allowed range $[0,1]$: 
Since $G(\Delta)$ increases strictly monotonously, $\tilde\nu$ decreases strictly monotonously with $\Delta$ in equation (\ref{eq:condition_cal_L}).
Consequently, $\tilde\nu\in[0,1]$ exactly if $\Delta_{\text{min}}\le\Delta<\frac\pi 2$ where $\Delta_{\text{min}}$ is defined by
\begin{equation}\label{eq:def_Delta_min}
    1 = \frac{1}{18G(\Delta_{\text{min}})}.
\end{equation}
Replacing $G(\Delta_{\text{min}})$ by its explicit term and solving for $\Delta_{\text{min}}$, this can be shown to be equivalent to
%\begin{align*}
%    G(\Delta_{min}) &= \frac{1}{18}\\
%    \left(\frac{1}{1+\cot(\Delta_{min}/4)}\right)^2 &= \frac 1{18}\\
%    \frac{1}{1+\cot(\Delta_{min}/4)} &= \frac 1{3\sqrt 2}\\
%    \cot(\Delta_{min}/4) &= 3\sqrt 2 - 1\\
%    \tan(\Delta_{min}/4) &= \frac{1}{3\sqrt 2 - 1} \approx 0.3084\\
%    \tan(\Delta_{min}/4) &= \frac{3\sqrt 2 + 1}{17} \approx 0.3084\\
%\end{align*}
\[\tan(\Delta_{min}/4) = \frac{1}{3\sqrt 2 - 1}\]
or
\[\Delta_{\text{min}}\approx 1.196\]
\begin{comment}
    For the present purposes, it will be sufficient to estimate $\Delta_{\text{min}}$ very roughly: Since the tangens function is strictly increasing on $(0,\frac\pi 2)$ and since a known value is 
    \[\tan (\frac\pi{3}/4) = 2-\sqrt 3 <  \frac{1}{3\sqrt 2 - 1}\]
    it is clear that
    \begin{equation}
        \Delta_{\text{min}} > \frac\pi 3.
    \end{equation}
\end{comment}
In summary, the level set
\begin{equation}
    \mathcal{L} = \{(\tilde\nu,\Delta)|  \tilde\nu = (18G(\Delta))^{-1}, \Delta\in[\Delta_{\text{min}}, \frac\pi 2)\}
\end{equation}
can be interpreted as the graph of a strictly monotonously decreasing function $\tilde\nu(\Delta)$ on $[\Delta_{\text{min}},\frac\pi 2)$ with
\[\lim_{\Delta\rightarrow\Delta_{\text{min}}}\tilde\nu(\Delta) = 1.\]
The situation is depicted in Figure \ref{fig:plot_B1_B2}.
Due to the monotonicity properties of $B_1$ it is now clear that $B_1(\tilde\nu, \Delta) > 0.81$ if there is a point $(\tilde\nu_0 ,\Delta_0) \in \mathcal{L}$ such that $\tilde\nu < \tilde\nu_0$ and $\Delta < \Delta_0$.

To prove Theorem \ref{theorem:main_result2} it is now sufficient to show that $B_2(\tilde\nu, \Delta) > 0.81$ in all other cases, \ite if $(\tilde\nu ,\Delta)$ is located in the hatched area of Figure \ref{fig:plot_B1_B2}. Due to the continuity and monotonicity of $\mathcal{L}$ one can then find $\tilde\nu_0 < \tilde\nu$ such that $(\tilde\nu_0, \Delta)\in\mathcal{L}$. The fact that $B_2$ increases strictly monotonously in $\nu$ then implies
\[B_2(\tilde\nu, \Delta) > B_2(\tilde\nu_0,\Delta).\]
Consequently, proving Theorem \ref{theorem:main_result2} is now reduced to showing that
\[B_2(\tilde\nu,\Delta)|_{\mathcal{L}} > 0.81\]

\begin{figure}
    \centering
    \begin{tikzpicture}[>=stealth, scale = 2]

% Draw the axes
\draw[->] (0,-0.2) -- (0, 2.2) node[above] {$\tilde\nu$};
\draw[->] (-0.2,0) -- (3.2, 0) node[below right] {$\Delta$};
\draw[dotted] (0, 2) -- (3, 2);
\draw[dotted] (3, 0) -- (3, 2);
\draw[dotted] (2.3, 0) -- (2.3, 2);

% line for B_2 text
\draw[->] (3.4, 1.4) -- (2.8, 1.8);

% axis labels
\draw (0,0) node[below left] {$0$};
\draw (2, 0.1) -- ( 2,-0.1) node[below] {$\frac\pi 3$};
\draw (2.3, 0.1) -- ( 2.3,-0.1) node[below] {$\Delta_{\text{min}}$};
\draw (3,0) -- (3,-0.1) node[below] {$\frac{\pi}{2}$};
\draw (0.1,  2) -- (-0.1,  2) node[left] {$1$};

% Level set L
\draw[pattern=north east lines, pattern color=black, draw=none, thick]
    (2.3, 2) 
    .. controls (2.5, 1.7) and (2.55, 1.65) .. (2.6, 1.6) % smooth curve
    -- (3, 1.3)                                          % sharp corner
    -- (3, 2)                                            % another corner
    -- cycle;                                            % close path
\draw[draw=\orange, ultra thick]
    (2.3, 2) 
    .. controls (2.5, 1.7) and (2.55, 1.65) .. (2.6, 1.6) % smooth curve
    -- (3, 1.3);                                          % sharp corner

% Texts
\draw (1,1) node[] {$B_1 > 0.81$};
\draw (3.6, 1.2) node[] {$B_2 > \inf B_2|_{\mathcal L}$};
\draw (2.65, 1.3) node[text=\orange] {$\mathcal{L}$};

\end{tikzpicture}
    \caption{To prove Theorem \ref{theorem:main_result2}, we have to show that at every point in $[0,1]\times(0,\frac\pi 2)$ at least one of the two functions $B_1$ and $B_2$ is greater than $0.81$. The monotonicity of $B_1$ guarantees this for $B_1$ below and to the left of the level set $\mathcal{L}$. The monotonicity of $B_2$ implies that no value of $B_2$ in the hatched area of the chart is larger than the infimum of $B_2$ restricted to $\mathcal{L}$.}
    \label{fig:plot_B1_B2}
\end{figure}
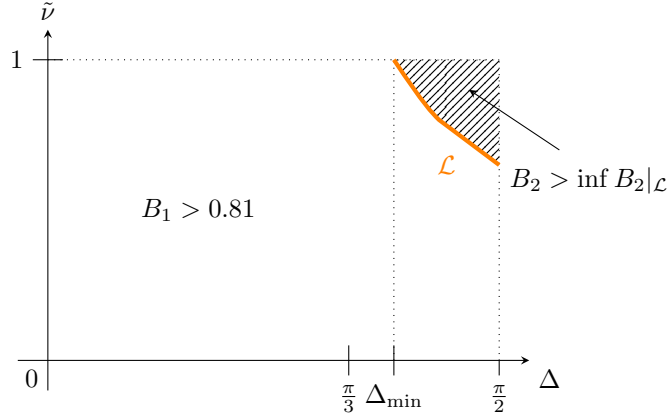

\subsection{Estimating the value of $B_2$ on $\mathcal L$}

Since all points on $\mathcal{L}$ satisfy (\ref{eq:condition_cal_L}), the restricted function $B_2(\tilde\nu,\Delta)|_{\mathcal{L}}$ can be identified with a function of $\Delta$ only:
\[B_{2,\mathcal{L}}(\Delta) := B_2(\frac{1}{18G(\Delta)},\Delta) \quad \text{ for }\Delta\in[\Delta_{\text{min}},\frac\pi 2).\]
Using the explicit expression (\ref{eq:def_B2}) for $B_2$ one arrives at
\begin{align}
    B_{2,\mathcal{L}}(\Delta)
    %&= 1 - \left(2 - \sec^2 \frac\Delta 2\right)\left( 1-\left(2+\frac{1}{18G(\Delta)} (2G(\Delta)-1)\right)^{-2}\right) \nonumber\\
    %&= 1 - \left(2 - \sec^2 \frac\Delta 2\right)\left( 1-\frac{1}{\left(2+\frac 19 (1-\frac{1}{2G(\Delta)})\right)^{2}}\right) \\
    &= 1 -  \left(2 - \sec^2 \frac\Delta 2\right)\left( 1-\left(\frac{19}{9}- \frac 1{18} G(\Delta)^{-1}\right)^{-2}\right) \label{eq:B2_on_L}
\end{align}
To prove Theorem \ref{theorem:main_result2} we need to show that $B_{2,\mathcal{L}}(\Delta) > 0.81$ on its domain. To that end we will replace the terms (\ref{eq:B2_on_L}) step by step with linear estimations in order to simplify the expression until it becomes analytically tractable.

\subsubsection*{Linearizing the secant term}

Write $F(\Delta)$ for the first bracket in (\ref{eq:B2_on_L}), \ite
\[F(\Delta) := 2 - \sec^2 \frac\Delta 2.\]
Its derivative is
\begin{align*}
    F'(\Delta) &= -\sec^2(\frac\Delta 2) \tan(\frac\Delta2).
\end{align*}
The function $F(\Delta)$ is concave on $(0,\frac\pi 2)$ and we can therefore estimate it from above by any of its tangents.
In particular, the choice $\Delta = \frac\pi 3$ (which is inspired by numerical studies of the problem) as the point of tangency yields for all $\Delta \in (0,\frac\pi 2)$ 
\begin{align}
    F(\Delta) &\le F'(\frac\pi 3)(\Delta-\frac\pi 3)+F(\frac\pi 3)\nonumber\\
    &= -\frac{4\sqrt 3}{9}(\Delta-\frac \pi 3) + \frac 23 \label{eq:estimate_F_tangent}
\end{align}

\subsubsection*{Linearizing the cotangent term in $G(\Delta)$}
Due to the convexity of the tangent function on the interval $[0,\frac\pi 8]$ the estimate
\[\tan\frac\Delta 4 \le \frac\Delta 4 \left(\frac{\tan(\frac\pi 8)}{\frac\pi 8}\right) = \frac 2\pi \Delta (\sqrt 2 - 1)\]
holds for $\Delta\in(0,\frac\pi 2)$. Putting this bound into the definition (\ref{eq:def_G}) of $G(\Delta)$ yields
\begin{align} \label{eq:estimate_G_tangent}
    G(\Delta)^{-1} 
    %\ge  \left(1 + \frac{1}{\frac 2\pi(\sqrt 2-1)\Delta}\right)^2\\
    &\ge \left(1 + \frac\pi 2(\sqrt 2+1)\frac 1 \Delta\right)^2
\end{align}

Before putting the estimates on $F$ and $G$ back into (\ref{eq:B2_on_L}) one needs to carefully check the signs of the different parts of the expression for $B_{2,\mathcal{L}}$. The first bracket in (\ref{eq:B2_on_L}), which we called $F(\Delta)$, is positive on $[\Delta_{\text{min}},\frac\pi 2)$. To see this, use standard trigonometric identities to rewrite $F(\Delta)$ as
\[F(\Delta) = \frac{\cos\Delta}{\cos^2(\frac \Delta2)}.\]

The term 
\[\frac{19}{9}-\frac{1}{18}G(\Delta)^{-1}\]
from the inner bracket of (\ref{eq:B2_on_L}) is also positive on $(\Delta_{\text{min}},\frac\pi 2)$. To see this, note that $G(\Delta_{\text{min}}) = \frac{1}{18}$ by definition (\ref{eq:def_Delta_min}) of $\Delta_{\text{min}}$, and keep in mind that $G(\Delta)$ increases with $\Delta$. 

In summary, when putting the estimates (\ref{eq:estimate_F_tangent}) and (\ref{eq:estimate_G_tangent}) into (\ref{eq:B2_on_L}), one obtains
\begin{align}
    B_{2,\mathcal{L}}(\Delta) &\ge 1 - \left(-\frac{4\sqrt 3}{9}(\Delta-\frac \pi 3) + \frac 23\right)\left( 1-\left(\frac{19}{9} - \frac{1}{18} \left(1+\frac k\Delta\right)^{2}\right)^{-2}\right)\nonumber\\
    %&= 1 - \left(-\frac{4\sqrt 3}{9}(\Delta-\frac \pi 3) + \frac 23\right)\left( 1-\frac{81}{\left(19 - \frac{1}{2} \left(1+\frac k\Delta\right)^2\right)^{2}}\right)\\
    %&= 1 - \left(-\frac{4\sqrt 3}{9}(\Delta-\frac \pi 3) + \frac 23\right)\left( 1-81\left(18\frac 12 - \frac k\Delta - \frac 12 \frac{k^2}{\Delta^2}\right)^{-2}\right) \\
    &= 1 - \left(-\frac{4\sqrt 3}{9}(\Delta-\frac \pi 3) + \frac 23\right)\left( 1-81H(\Delta)^{-2}\right) \label{eq:B2_on_L2}
\end{align}
with
\[k := \frac\pi 2(\sqrt 2 +1)\]
and
\[H(\Delta) := \frac{37}2 - \frac k\Delta - \frac 12 \frac{k^2}{\Delta^2}.\]

\subsubsection*{Linearizing $H(\Delta)$}
Being the sum of concave functions on $[\Delta_{\text{min}},\frac\pi 2)$, $H(\Delta)$ is concave on that interval. Consequently, it can be estimated from above by any of its tangents. Choosing $\Delta = k/3$ as the point of tangency yields the bound
\[H(\Delta) \le D'(\frac k3)(\Delta - \frac k3) + D(\frac k3)\]
which can be simplified in a straightforward calculation to 
%\begin{align*}
%    H(\Delta) &\le H'(\frac k3)(\Delta - \frac k3) + H(\frac k3)\\
%    &= (\frac 9 k + \frac {27} k)(\Delta - \frac k 3) + 18\frac 12 - 3 - %\frac 9 2\\
%    &= \frac{36}k(\Delta - \frac k 3) + 11\\
%    &= \frac{36}k\Delta - 12 + 11\\
%    &= \frac{36}k\Delta - 1\\
%\end{align*}
\begin{equation} \label{eq:estimate_D_tangent}
    H(\Delta) \le \frac{36}k\Delta - 1
\end{equation}
Putting this estimate back into (\ref{eq:B2_on_L2}) one obtains
\begin{align} \label{eq:B2_on_L3}
    B_{2,\mathcal{L}}(\Delta) &\ge 1 - \left(-\frac{4\sqrt 3}{9}(\Delta-\frac \pi 3) + \frac 23\right)\left( 1-\frac{81}{\left(\frac{36}k\Delta - 1\right)^{2}}\right)
\end{align}

\subsubsection*{Solving the third degree polynomial}

In the following, call $f(\Delta)$ the r.h.s. of (\ref{eq:B2_on_L3}).
Introduce the auxiliary variable
\[D = \frac{36}{k}\Delta - 1 \quad \leftrightarrow \quad \Delta = \frac{k}{36}(D+1)\]
and let $\tilde f(D)$ be the transformed function of $f(\Delta)$, \ite
\[\tilde f(D) = f(\Delta) \quad \text{ if } D = \frac{36}{k}\Delta - 1\]
Then $\tilde f(D)$ is defined for
\[D\in\left[\frac{36}{k}\Delta_{\text{min}}-1, \frac{36\pi}{2k}-1\right).\]
and its infimum on that interval is the same as the infimum of $f(\Delta)$ on $[\Delta_{\text{min}}, \frac\pi 2)$.
A short calculation shows that
\begin{align}
    \tilde f(D) %&=  1 - \left(-\frac{4\sqrt 3}{9}(\frac{k}{36}(D+1)-\frac \pi 3) + \frac 23\right)\left( 1-\frac{81}{D^2}\right)\\
     %&=  1 - \left(-\frac{4\sqrt 3}{9}(\frac{k}{36}D+\frac{k}{36}-\frac \pi 3) + \frac 23\right)\left( 1-\frac{81}{D^2}\right)\\
     &= 1 + \left(\frac{\sqrt 3k}{81}D+\frac{\sqrt 3k}{81}-\frac{4\sqrt 3\pi}{27} - \frac 23\right)\left( 1-\frac{81}{D^2}\right).
\end{align}
Observe that $\tilde f$ is continuously differentiable on its domain. To locate its extrema, set $\tilde f'(D) = 0$. A short calculation shows that this condition is equivalent to
\begin{equation} \label{eq:third_deg_polynomial}
    0 = D^3 + p D + q 
\end{equation}
with the shorthand notation
\[p  = 81, \quad q = 162 - 162\frac{36}{k}(\frac\pi 3 + \frac{\sqrt 3}{2}).\]
\begin{comment}
    \begin{align}
        f'(\Delta) &= \frac{\sqrt 3k}{81}\left( 1-\frac{81}{D^2}\right) + 2\left(\frac{\sqrt 3k}{81}D+\frac{\sqrt 3k}{81}-\frac{4\sqrt 3\pi}{27} - \frac 23\right)\frac{81}{D^3}\\
        &= \frac{\sqrt 3k}{81} -\frac{\sqrt 3 k}{D^2} + \left(\sqrt 3kD+\sqrt 3k-12\sqrt 3\pi - 54\right)\frac{2}{D^3}\\
    \end{align}
    Set $f'(\Delta) = 0$ and multiply with $D^3$:
    \begin{align}
        0 &= \frac{\sqrt 3k}{81}D^3 - \sqrt 3 k D + 2\sqrt 3kD+2\sqrt 3k-24\sqrt 3\pi - 108\\
        &= \frac{\sqrt 3k}{81}D^3 + \sqrt 3kD+2\sqrt 3k-24\sqrt 3\pi - 108\\
        &= D^3 + 81 D+ 162 -24*81\pi \frac 1k  - 108*81*\frac{1}{\sqrt 3 k}\\
        &= D^3 + 81 D+ 162 - 162(12\pi \frac 1k  + 54\frac{1}{\sqrt 3 k})\\
        &= D^3 + 81 D+ 162 - 162\frac{36}{k}(\frac\pi 3 + \frac{\sqrt 3}{2})\\
        &= D^3 + p D + q
    \end{align}    
\end{comment}
The cubic discriminant 
\[-4p^3 - 27q^2\]
of the depressed cubic polynomial (\ref{eq:third_deg_polynomial}) is negative, so there is exactly one real root of the polynomial.
Cardano's formula tells us that the real root of the polynomial is
\[D_0 = \sqrt[3]{u_1}+\sqrt[3]{u_2}\]
where
\begin{align*}
    u_{1/2} &= -\frac q2 \pm\sqrt{\frac{q^2}4+\frac{p^3}{27}}.
\end{align*}
In summary, the value $D_0$ is the only critical point of $\tilde f(D)$. Since $f(D)$ (or, strictly speaking, its natural extension to the domain $\R^+$) approaches infinity both for $D \rightarrow 0$ and $D \rightarrow \infty$, the critical point $D_0$ must be a minimum.
The corresponding value of the variable $\Delta$ is
\[\Delta_0 = \frac{k}{36}(D_0 + 1) \approx 1.386\]
and since this is in the domain $[\Delta_{\text{min}},\frac\pi2)$ of $f$ one may conclude that $\Delta_0$ is a minimum of $f$.
The numeric value that $f$ attains at its minimum can be computed directly and so we arrive at
\[B_{2,\mathcal{L}}(\Delta) \ge f(\Delta_0) \approx 0.8166 > 0.81,\]
thus concluding the proof of Theorem \ref{theorem:main_result2}.

%------------------------------------------
\section{Proof of Lemma \ref{lem:lower_bounb_variation_f}} \label{sec:proof_estimate_f}
%------------------------------------------

\subsection{Overview and main ideas}

The basic idea of the proof is that every permissible function $\hat f$ can be modified in a sequence of steps such that the total variation $V(\cdot)$ does not increase, the Fourier condition from Lemma \ref{lem:lower_bounb_variation_f} remains intact, and one always arrives at the minimizer that yields equality in (\ref{eq:lower_bound_variation_f}). The steps to modify $\hat f$ are:
\begin{enumerate}
    \item Construct the function $f_1$ by `sorting and canceling' positive and negative parts of $\hat f$ on $(0, \frac\pi 2)$.
    \item Construct the step function $f_2$ based on $f_1$ which has one positive, one negative, and one zero plateau in $(\tau_1,\tau_2)$.
    \item Construct a further simplified step function $f_3$ without the zero plateau.
    \item Choose the function $f_4$ such that it minimizes $V(\cdot)$ out of all possible functions $f_3$ that could result from the previous steps.
\end{enumerate}
Note that there is no need for the functions $f_1$, $f_2$, $f_3$ and $f_4$ to satisfy the conditions of Lemma \ref{lem:lower_bounb_variation_f} as $\hat f$ does - in fact, they will not. But we do impose the Fourier condition on each of them since otherwise we could only prove a weaker bound on $V(\hat f)$

\subsection{Notation and Preliminaries}
Let $\hat f$ be as in Lemma \ref{lem:lower_bounb_variation_f} and define
\begin{align*}
    \Omega_+ :=& \left\{t\in\left(0,\frac\pi 2\right)| \hat f(t) \ge 0\right\},\\
    \Omega_- :=& \left\{t\in\left(0,\frac\pi 2\right)| \hat f(t) < 0\right\},\\
    \Omega_r :=& \left[\frac{\pi}{2}, \pi\right].
\end{align*}
Note that both $\Omega_+$ and $\Omega_-$ are of non-zero measure since otherwise the Fourier condition of Lemma \ref{lem:lower_bounb_variation_f} cannot hold.

Let $\hat f_+, \hat f_-: \Omega \rightarrow \R$ be the positive and negative part of $\hat f$ on $(0,\frac\pi2)$, respectively, \ite 
\begin{equation*}
    \hat f_\pm(t) := \hat f(t)\, \chi_{\Omega_\pm}(t)
\end{equation*}
where $\chi$ is the characteristic function on the respective set. Find a point $t_x \in (0,\frac{\pi}{2})$ such that
\begin{equation} \label{eq:t_x}
    -\int_0^{t_x} \hat f_-(t) \dt = \int_{t_x}^{\frac\pi2} \hat f_+(t) \dt
\end{equation}
%\begin{equation} \label{eq:t_x_new}
%    \int_{\Omega_-\cap[0,t_x]} |\hat f(t)| \dt = \int_{\Omega_+\cap[t_x, \frac{\pi}{2}]} \hat f(t) \dt
%\end{equation}
The existence of at least one such $t_x$ is guaranteed by the intermediate value theorem since both sides of (\ref{eq:t_x}) are continuous functions in $t_x$ and the l.h.s. is zero for $t_x = 0$ and increases mono\-tonously in $t_x$ while the r.h.s. is zero for $t_x = \frac{\pi}{2}$ and decreases mono\-tonously.

In fact, $t_x$ can be more narrowly localized, namely $t_x \in (\tau_1,\tau_2)$.
To see that $\tau_1<t_x$ note that $\hat f$ is strictly negative on $[0, \tau_1)$. So if $t_x$ were smaller or equal $\tau_1$ one would have
\begin{equation*}
    -\int_0^{\tau_1} \hat f(t) \dt \ge -\int_0^{t_x} \hat f_-(t) \dt = \int_{t_x}^{\frac\pi2} \hat f_+(t) \dt = \int_{\Omega_+} \hat f(t) \dt.
\end{equation*} 
Given the monotonicity of the cosine function this would imply
\begin{equation*}
    -\int_0^{\tau_1} \hat f(t) \cos t \dt > \int_{\Omega_+} \hat f(t) \cos t \dt
\end{equation*}
and thus
\begin{align*}
    \int_\Omega \hat f(t) \cos t \dt &= \int_{\Omega_+\cup\Omega_-\cup\Omega_r} \hat f(t) \cos t \dt\\
    &< \int_{(\Omega_- \backslash[0,\tau_1])\cup\Omega_r} \hat f(t) \cos t \dt <  0,    
\end{align*}
which is a contradiction to the assumptions on $\hat f$.

In a similar way, one can show that $t_x < \tau_2$: The function $\hat f$ is strictly positive on $(\tau_2, \frac{\pi}{2}]$, so if $t_x$ were greater or equal $\tau_2$ one would have
\begin{equation*}
    \int_{\tau_2}^\frac{\pi}{2} \hat f(t) \dt \ge \int_{t_x}^\frac{\pi}{2} \hat f_+(t) \dt = -\int_0^{t_x} \hat f_-(t) \dt = -\int_{\Omega_-} \hat f(t) \dt
\end{equation*} 
so that
\[\int_\Omega \hat f(t) \sin t \dt >  \int_{(\Omega_+\backslash[\tau_2, \frac\pi 2])\cup\Omega_r} \hat f(t) \sin t \dt > 0,\]
which is again a contradiction to the assumptions on $\hat f$.

\subsection{Constructing $f_1$}

Now construct the function $f_1$ with two initially unknown parameters
\[\nu_1, \nu_2 \in [0,1]\]
as follows:
\begin{equation} \label{eq:def_f1_new}
    f_1(t) := 
    \begin{cases}    
        \nu_1 \hat f_+(t) & \text{for } t \in (0, t_x] \\[10pt]
        \nu_2 \hat f_-(t) & \text{for } t \in (t_x, \frac{\pi}{2}) \\[10pt]
        \hat f(t) & \text{for } t \in [\frac{\pi}{2}, \pi] \\[10pt]
        -f_1(t-\pi) & \text{for } t \in (\pi, 2\pi]
    \end{cases}
\end{equation}

\begin{figure}
    \centering
    \begin{tikzpicture}[>=stealth, scale = 2]

\def\miu{1.5}

% Draw the axes
\draw[->] (0,-1.5) -- (0, 1.7) node[above] {$\hat f$, $f_1$};
\draw[->] (-0.2,0) -- ({1.1*pi},0) node[below right] {$t$};
\draw[dotted] (0.2,  -1) -- (0.2  ,-1.3);
\draw[dotted] (0.7  ,-1) -- (0.7  ,-1.5);
\draw[dotted] (1.3  ,-1) -- (1.3  ,-1.3);
\draw[dotted] (1.55 ,-1) -- (1.55 ,-1.3);
\draw[dotted] (3.15 ,-1) -- (3.15 ,-1.5);

% horizontal line for mu, alpha
\draw[dotted] (0 , \miu) -- (3.15 ,\miu);
\draw[dotted] (0 , 1) -- (3.15 , 1);

% axis labels
\draw (0,0) node[below left] {$0$};
\draw (0.2, 0.1) -- ( 0.2,-0.1) node[below right] {$\tau_1$};
\draw (1.0,-0.1) -- (   1, 0.1) node[above] {$t_x$};
\draw (1.3, 0.1) -- ( 1.3,-0.1) node[below] {$\tau_2$};
\draw ({pi/2},0) -- ({pi/2},-0.1) node[below] {$\frac{\pi}{2}$};
\draw ({pi},  0) -- ({pi}  ,-0.1) node[below] {$\pi$};
\draw (0.1,  \miu) -- (-0.1,  \miu) node[left] {$\delta + \nu$};
\draw (0.1,  1) -- (-0.1,  1) node[left] {$\delta$};
\draw (0.1, -1) -- (-0.1, -1) node[left] {$-\delta$};
\draw (0.15, -1.3) node[below] {$\Omega_-$};
\draw (0.45, -1.3) node[below] {$\Omega_+$};
\draw (1.0,  -1.3) node[below] {$\Omega_-$};
\draw (1.45, -1.3) node[below] {$\Omega_+$};
\draw (2.3,  -1.3) node[below] {$\Omega_r$};

% Function \hat f
\draw[\orange,thick]
  plot[smooth] coordinates {
    (-0.1,     -1.2)
    ( 0.0,     -1.2)
    ( 0.05,    -1.1)
    ( 0.1,     -0.8)
    ( 0.15,    -0.3)
    ( 0.2,      0.3)
    ( 0.25,     0.9)
    ( 0.35,     1.2)
    ( 0.45,     1.3)
    ( 0.55,     1.2)
    ( 0.65,     0.7)
    ( 0.7,      0.0)
    ( 0.8,     -0.5)
    ( 0.9,     -0.8)
    ( 1.0,     -0.9)
    ( 1.1,     -0.8)
    ( 1.2,     -0.5)
    ( 1.3,      0.0)
    ( 1.4,      0.6)
    ( 1.5,      1.1)
    ( 1.6,      1.3)
    ( 1.7,      1.4)
    ( 1.9,      1.5)
    ( 2.2,      1.5)
    ( 2.5,      1.4)
    ( 2.7,      1.3)
    ( 3.0,      1.2)
    ( 3.15,     1.15)
    ( 3.2,      1.1)
    ( 3.25,      0.8)
  };

% Function f_1
\draw[\blue,dotted, ultra thick]
  plot coordinates {
    (-0.1,     -1.2)
    ( 0.0,  -1.2)
    ( 0.0,   0.0)
    ( 0.17,     0.0 * 0.5)
    ( 0.2,      0.3 * 0.5)
    ( 0.25,     0.9 * 0.5)
    ( 0.35,     1.2 * 0.5)
    ( 0.45,     1.3 * 0.5)
    ( 0.55,     1.2 * 0.5)
    ( 0.65,     0.7 * 0.5)
    ( 0.7,      0.0)
    ( 1.0,      0.0)
    ( 1.0,     -0.9 * 0.9)
    ( 1.1,     -0.8 * 0.9)
    ( 1.2,     -0.5 * 0.9)
    ( 1.3,      0.0)
    ( 1.3,      0.0)
    ( 1.55,     0.0)
    ( 1.55,     1.2)
    ( 1.6,      1.3)
    ( 1.7,      1.4)
    ( 1.9,      1.5)
    ( 2.2,      1.5)
    ( 2.5,      1.4)
    ( 2.7,      1.3)
    ( 3.0,      1.2)
    ( 3.15,     1.15)
    ( 3.15,     0.0)
    ( 3.25,     0.0)
  };

% Fill the area between the horizontal line f=alpha (for t in [pi/2, pi])
% and the t-axis (f=0) with a hatching pattern 

\fill[pattern=north east lines,draw=none]
  ({pi/2},0) -- ({pi},0) -- ({pi},1) -- ({pi/2},1) -- cycle;

\end{tikzpicture}
    \caption{The function $f_1$\ifcolor (blue dotted line) \else (dotted line) \fi is constructed from the function $\hat f$ \ifcolor (orange line)\else (solid line)\fi. Here it is shown for $\nu_1 = 0.5$ and $\nu_2 = 0.9$.}
    \label{fig:plot_f_f1}
\end{figure}
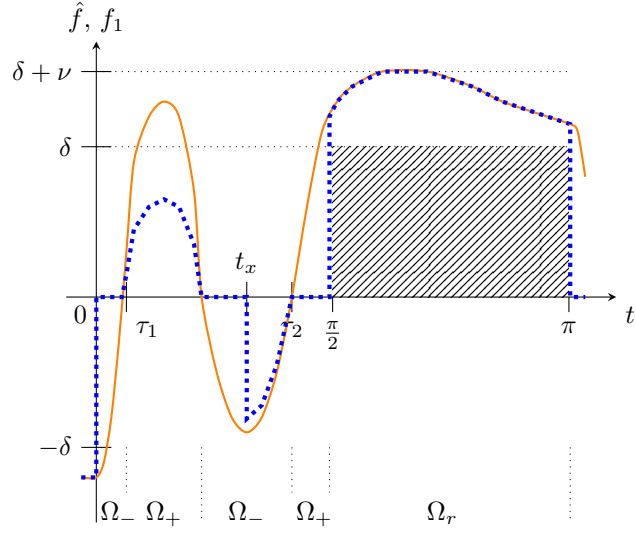

Compare the total variation of $\hat f$ and of $f_1$ on the following  partitioning of $\Omega$ to confirm that the total variation has not increased:
\begin{equation}
    \begin{split} \label{eq:partitioning}
        V(f_1) = &\, V_{[0,\tau_1]}(f_1) + V_{[\tau_1, \tau_2]}(f_1) +  V_{[\tau_2,\frac{\pi}{2}]}(f_1)\\
        &+ V_{[\frac{\pi}{2},\pi]}(f_1) + V_{[\pi,2\pi]}(f_1)
    \end{split}
\end{equation}
As for the first summand, 
\[V_{[0,\tau_1]}(f_1) = |\hat f(0)|\le V_{[0,\tau_l]}(\hat f).\]

The second summand can be estimated by partitioning the interval $[\tau_1, \tau_2]$ further into intervals between subsequent zeros of $\hat f$.
On each such interval it is apparent that the total variation of $f_1$ is bounded from above by the total variation of $\hat f$. In fact, as will become clear below, $\nu_1$ and $\nu_2$ must be strictly smaller than one and therefore the following inequality is strict: 
\[V_{[\tau_1, \tau_2]}(f_1) < V_{[\tau_1, \tau_2]}(\hat f).\]

The third summand in (\ref{eq:partitioning}) is
\[V_{[\tau_2,\frac{\pi}{2}]}(f_1) = |\hat f(\frac{\pi}{2})| \le V_{[\tau_2,\frac{\pi}{2}]}(\hat f).\]

The fourth summand is trivially identical to $V_{[\frac{\pi}{2},\pi]}(\hat f)$ and finally the last summand is equal to the sum of the first four due to the anti-periodicity of $f_1$ and $\hat f$. So in summary it is established that $V(f_1) < V(\hat f)$.

To conclude step 1 it remains to show that $\nu_1$ and $\nu_2$ can be chosen such that $f_1$, just like $\hat f$, satisfies the Fourier condition of Lemma \ref{lem:lower_bounb_variation_f}, \ite, that the integrals with $\sin t$ and $\cos t$ vanish.

The parameter space of $\nu_1$ and $\nu_2$ can be identified with the square $\Sigma =[0,1] \times [0,1]$ in $\R^2$ (cf. Figure \ref{fig:plot_square}). Let
\begin{equation}
    \begin{aligned}
        s_1(\nu_1, \nu_2) &:= \int_\Omega f_1(t) \sin t \dt \quad \text{ and} \\
        c_1(\nu_1, \nu_2) &:= \int_\Omega f_1(t) \cos t \dt.
    \end{aligned}
\end{equation}
The Fourier condition of Lemma \ref{lem:lower_bounb_variation_f} is met where $s_1$ and $c_1$ vanish. The functions $s_1$ and $c_1$ are linear and non-constant on $\Sigma$ and therefore their zeros form straight lines. By comparing the values of $s_1$ and $c_1$ at the corners of the square $\Sigma$ one can prove that these lines must intersect within the square and therefore a permissible choice of $\nu_1$ and $\nu_2$ exists which satisfies the Fourier condition.

Lower left corner: For $\nu_1 = \nu_2 = 0$ putting the definition of $f_1$ into the definitions of $s_1$ and $c_1$ directly shows that $s_1(0,0) > 0$ and $c_1(0,0) < 0$. % In $(0,\pi$ the function $f_1$ is nonzero only in $[\frac{\pi}{2}, \pi]$ where it is equal to $\hat f$ and therefore positive. Integration over its product with sine is therefore positive and with cosine negative because cos(t) < 0 for t > $\frac{\pi}{2}$.

Upper right corner: For $\nu_1 = \nu_2 = 1$ the functions $\hat f$ and $f_1$ only differ by $\hat f_-$ on $(0,t_x)$ and by $\hat f_+$ on $(t_x, \frac{\pi}{2})$ and by the corresponding terms on $(\pi, \frac 32 \pi)$ due to the anti-periodic nature of both functions. Thus
\begin{equation} \label{eq:balancing_cond_new}
    \begin{aligned}
        0 &= \int_\Omega\hat f(t) \sin t \dt  \\
        &= s_1(\nu_1,\nu_2) + 2\int_0^{t_x} \hat f_-(t) \sin t \dt + 2\int_{t_x}^{\frac\pi 2} \hat f_+(t) \sin t \dt, \\
        0 &= \int_\Omega\hat f(t) \cos t \dt \\
        &= c_1(\nu_1,\nu_2) + 2\int_0^{t_x} \hat f_-(t) \cos t \dt + 2\int_{t_x}^{\frac\pi 2} \hat f_+(t) \cos t \dt 
    \end{aligned}
\end{equation}
Based on (\ref{eq:t_x}) and the monotonicity properties of the sine and cosine functions in $(0, \frac{\pi}{2})$ one can conclude that
\begin{equation*}
    - \int_{0}^{t_x} \hat f_-(t) \sin t \dt < \int_{t_x}^{\frac\pi 2} \hat f_+(t) \sin t \dt
\end{equation*}
and
\begin{equation*}
    - \int_{0}^{t_x} \hat f_-(t) \cos t \dt > \int_{t_x}^{\frac\pi 2} \hat f_+(t) \cos t \dt.
\end{equation*}
Putting these into (\ref{eq:balancing_cond_new}) one can see that $s_1(1,1)<0$ and $c_1(1,1) > 0$.

Upper left corner: Since the value of $f_1$ on $(0,t_x)$ increases with $\nu_1$, one can estimate $s_1(0,1) < s_1(1,1) < 0$. Similarly, since increasing $\nu_2$ reduces the value of $f_1$ on $(t_x,\frac{\pi}{2})$, one directly gets $c_1(0,1) < c_1(0,0) < 0$.

Lower right corner: Analogous arguments show that $s_1(1,0) > s_1(0,0) > 0$ and $c_1(1,0) > c_1(1,1) > 0$.

The results from the four corners show that the zero line of $s_1$ must cross the square $\Sigma$ from left to right, while the zero line of $c_1$ crosses the square from top to bottom. Therefore an intersection between both lines must exist in $\Sigma$, corresponding to a choice of $\nu_1$ and $\nu_2$ where the Fourier condition in Lemma \ref{lem:lower_bounb_variation_f} is met.

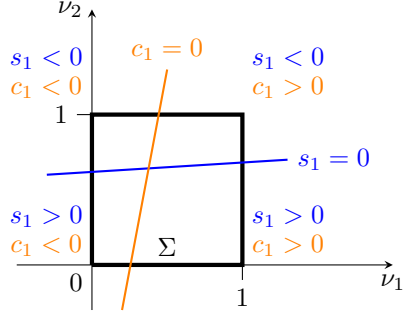
\begin{figure}
    \centering
    \begin{tikzpicture}[>=stealth, scale = 2]

\def\e{0.0}  % distance between both functions

% Draw the axes
\draw[->] (0,   -0.3) -- (0,  1.7) node[left] {$\nu_2$};
\draw[->] (-0.5,   0) -- (2.0,  0) node[below] {$\nu_1$};

% axis labels
\draw (0,0) node[below left] {$0$};
\draw (1.0, 0.1) -- (   1,-0.1) node[below] {$1$};
\draw (0.1,   1) -- (-0.1,   1) node[left] {$1$};

% square
\draw[black, ultra thick]
  plot coordinates {
    ( 0.0,      0.0)
    ( 0.0,      1.0)
    ( 1.0,      1.0)
    ( 1.0,      0.0)
    ( 0.0,      0.0)
  };
\draw (0.5, 0.0) node[above] {$\Sigma$};
  
% horizontal line
\draw[\blue, thick] (-0.3,   0.6) -- ( 1.3,   0.7) node[right] {$s_1 = 0$};

% vertikal line
\draw[\orange, thick] (0.2,   -0.3) --  (0.5,    1.3) node[above] {$c_1 = 0$};

%  s, c vs. 0
\draw (0,0) node[above left] {\shortstack{\textcolor{\blue}{$s_1 > 0$} \\ \textcolor{\orange}{$c_1 < 0$}}};
\draw (1,0) node[above right] {\shortstack{\textcolor{\blue}{$s_1 > 0$} \\ \textcolor{\orange}{$c_1 > 0$}}};
\draw (0,1.05) node[above left] {\shortstack{\textcolor{\blue}{$s_1 < 0$} \\ \textcolor{\orange}{$c_1 < 0$}}};
\draw (1,1.05) node[above right] {\shortstack{\textcolor{\blue}{$s_1 < 0$} \\ \textcolor{\orange}{$c_1 > 0$}}};

\end{tikzpicture}
    \caption{Estimating the values of $s_1(\nu_1,\nu_2)$ and $c_1(\nu_1,\nu_2)$ at the four corners of the square $\Sigma$ shows that the zero lines of $s_1$ and $c_1$ cross inside the square. The intersection corresponds to a choice of $\nu_1$ and $\nu_2$ which `balances' $f_1$ with respect to the Fourier condition.}
    \label{fig:plot_square}
\end{figure} 

\subsection{Constructing $f_2$}

Based on $f_1$ as constructed above with $\nu_1$ and $\nu_2$ chosen such that the Fourier condition is fulfilled, define the step function 
\begin{equation} \label{eq:def_f2}
    \begin{aligned}
        f_2(t) &:= \nu_3 \left(\sup_{t'\in[\tau_1,t_x]} f_1(t')\right)\chi_{(\tau_1,l]}(t) \\
        &+ \nu_4\left(\inf_{t'\in[t_x, \tau_2]} f_1(t')\right) \chi_{[r,\tau_2)}(t)\\
        &+ (\delta + \nu) \chi_{\frac{\pi}{2}}(t) + \delta \chi_{(\frac{\pi}{2}, \pi]}(t)      
    \end{aligned}
\end{equation}
for $t \in (0,\pi]$ and $f_2(t) = - f_2(t-\pi)$ for $t \in (\pi, 2\pi]$. Here, $\chi$ denotes the characteristic function on the respective point or interval, 
$\nu_3, \nu_4 \in [0,1]$ are two initially unknown parameters, and
\begin{align*}
    l :=\;& \tau_1 + \frac{\int_{\tau_1}^{t_x} f_1(t) \dt}{\sup_{t\in[\tau_1,t_x]} f_1(t)}  < t_x,\\
    r :=\;& \tau_2 -\frac{\int_{t_x}^{\tau_2} f_1(t) \dt}{\inf_{t\in[t_x, \tau_2]} f_1(t)} > t_x
\end{align*}
It is straightforward to check that the total variation has not increased, \ite, 
\[V(f_2) \le V(f_1).\]

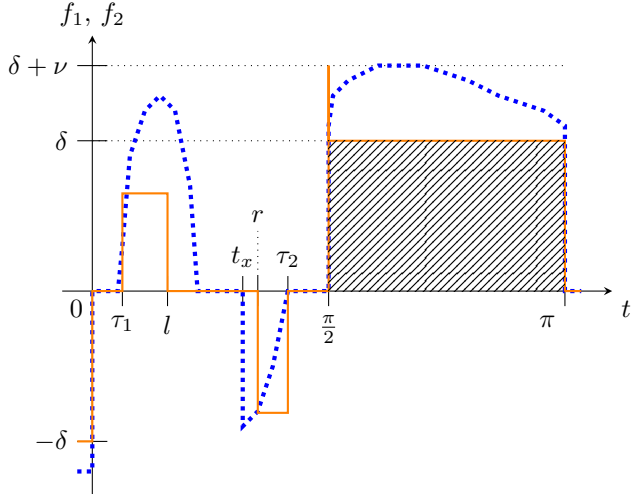
\begin{figure}
    \centering
    \begin{tikzpicture}[>=stealth, scale = 2]

\def\l{0.5}
\def\r{1.1}
\def\miu{1.5}

% Draw the axes
\draw[->] (0,-1.35) -- (0, 1.7) node[above] {$f_1$, $f_2$};
\draw[->] (-0.2,0) -- ({1.1*pi},0) node[below right] {$t$};

% horizontal line for mu, alpha
\draw[dotted] (0 , \miu) -- (3.15 ,\miu);
\draw[dotted] (0 , 1) -- (3.15 , 1);

% axis labels
\draw (0,0) node[below left] {$0$};
\draw (0.2, 0.1) -- ( 0.2,-0.1) node[below] {$\tau_1$};
\draw (\l,  0.1) -- (  \l,-0.1) node[below] {$l$};
\draw (1.0,-0.1) -- (   1, 0.1) node[above] {$t_x$};
\draw (\r, -0.1) -- (  \r, 0.1);
\draw [dotted] (\r,  0.1) -- (  \r, 0.4) node[above] {$r$};
\draw (1.3,-0.1) -- ( 1.3, 0.1) node[above] {$\tau_2$};
\draw ({pi/2},0) -- ({pi/2},-0.1) node[below] {$\frac{\pi}{2}$};
\draw ({pi},  0) -- ({pi}  ,-0.1) node[below left] {$\pi$};
%\draw (0,1) node[left] {$\delta$};
% Add tick mark at y=1
\draw (0.1,  \miu) -- (-0.1,  \miu) node[left] {$\delta+\nu$};
\draw (0.1,  1) -- (-0.1,  1) node[left] {$\delta$};
\draw (0.1, -1) -- (-0.1, -1) node[left] {$-\delta$};

% Function f_1
\draw[\blue,dotted, ultra thick]
  plot coordinates {
    (-0.1,     -1.2)
    ( 0.0,     -1.2)
    ( 0.0,      0.0)
    ( 0.17,     0.0)
    ( 0.2,      0.3)
    ( 0.25,     0.9)
    ( 0.35,     1.2)
    ( 0.45,     1.3)
    ( 0.55,     1.2)
    ( 0.65,     0.7)
    ( 0.7,      0.0)
    ( 1.0,      0.0)
    ( 1.0,     -0.9)
    ( 1.1,     -0.8)
    ( 1.2,     -0.5)
    ( 1.3,      0.0)
    ( {pi/2},   0.0)
    ( {pi/2},   1.1)
    ( 1.6,      1.3)
    ( 1.7,      1.4)
    ( 1.9,      1.5)
    ( 2.2,      1.5)
    ( 2.5,      1.4)
    ( 2.7,      1.3)
    ( 3.0,      1.2)
    ( {pi},     1.1)
    ( {pi},     0.0)
    ( 3.25,     0.0)
  };

% Function f_2
\draw[\orange, thick]
  plot coordinates {
    (-0.1,     -1.0)
    ( 0.0,     -1.0)
    ( 0.0,      0.0)
    ( 0.2,      0.0)
    ( 0.2,      1.3 * 0.5)
    ( \l,       1.3 * 0.5)
    ( \l,       0.0)
    ( \r,       0.0)
    ( \r,      -0.9 * 0.9)
    ( 1.3,     -0.9 * 0.9)
    ( 1.3,      0.0)
    ( {pi/2},   0.0)
    ( {pi/2},   \miu)
    ( {pi/2},   1)
    ( {pi},     1.0)
    ( {pi},     0.0)
    ( 3.25,     0.0)
  };

% 5) Fill the area between the horizontal line f=alpha (for t in [pi/2, pi])
%    and the t-axis (f=0) with a hatching pattern 

\fill[pattern=north east lines,draw=none]
  ({pi/2},0) -- ({pi},0) -- ({pi},1) -- ({pi/2},1) -- cycle;

\end{tikzpicture}
    \caption{The step function $f_2$ \ifcolor (orange line) \else (solid line) \fi is constructed from $f_1$ \ifcolor (blue dotted line)\else (dotted line)\fi , here shown for $\nu_3 = 0.5$ and $ \nu_4 = 0.9$.}
    \label{fig:plot_f1_f2}
\end{figure}

To conclude step 2 it remains to show that $\nu_3$ and $\nu_4$ can be chosen such that $f_2$, just like $\hat f$ and $f_1$, satisfies the Fourier condition. The proof is analogous to the one for $f_1$:

The parameter space of $\nu_3$ and $\nu_4$ can be identified with the square $\Sigma=[0,1]$. Let
\begin{equation*}
    s_2(\nu_3, \nu_4) := \int_\Omega f_2(t) \sin t \dt 
\end{equation*}
and
\begin{equation*}
    c_2(\nu_3, \nu_4) := \int_\Omega f_2(t) \cos t \dt.
\end{equation*}
Again, one can argue that two straight lines of zeros of $s_2$ and $c_2$ must intersect in the square $\Sigma$ by analyzing the signs of $s_2$ and $c_2$ in the corners:

Lower left corner: For $\nu_3 = \nu_4 = 0$ putting the definition of $f_2$ into the definitions of $s_2$ and $c_2$ directly shows that $s_2(0,0) > 0$ and $c_2(0,0) < 0$. % In $(0,\pi)$ the function $f_2$ is nonzero only in $[\frac{\pi}{2}, \pi]$ where it is equal to $\delta$ and therefore positive. Integration over its product with sine is therefore positive and with cosine negative because cos(t) < 0 for t > $\frac{\pi}{2}$.

Upper right corner: For $\nu_1 = \nu_2 = 1$ one has
\begin{equation} \label{eq:s_2}
    \begin{split}
    s_2(1,1) =& s_1(1,1) + \int_\Omega \left(f_2(t)-f_1(t)\right) \sin t \dt \\
    =& s_1(1,1) + 2\left(\int_{\tau_1}^{t_x} f_2(t) \sin t \dt - \int_{\tau_1}^{t_x} f_1(t) \sin t \dt\right) \\
    &+ 2\left(\int_{t_x}^{\tau_2} f_2(t) \sin t \dt - \int_{t_x}^{\tau_2} f_1(t) \sin t \dt\right) \\
    &+ 2\left(\int_{\frac\pi 2}^\pi f_2(t) \sin t \dt - \int_{\frac\pi 2}^\pi f_1(t) \sin t \dt\right)
    \end{split}
\end{equation}
Here the first summand $s_1(1,1)$ is already known to be negative and one can use the monotonicity properties of sine and cosine and the fact that
\begin{eqnarray*}
    \int_{\tau_1}^{t_x} f_1(t) \dt &=& \int_{\tau_1}^{t_x} f_2(t) \dt,\\
    \int_{t_x}^{\tau_2} f_1(t) \dt &=& \int_{t_x}^{\tau_2} f_2(t) \dt
\end{eqnarray*}
to see that the three other summands in (\ref{eq:s_2}) are also non-positive. Thus $s_2(1,1) < 0$.
An analogous argument shows that $c_2(1,1) > 0$.

Upper left corner: Since the value of $f_2$ on $(0, t_x)$ increases with $\nu_3$, one can estimate
\[s_2(0,1) < s_2(1,1) < 0.\]
Since increasing $\nu_4$ reduces the value of $f_2$ on $(t_x,\tau)$, one also knows that
\[c_2(0,1) < c_2(0,0) < 0.\]

Lower right corner: Analogous arguments show that
\[s_2(1,0) > s_2(0,0) > 0\] 
and 
\[c_2(1,0) > c_2(1,1) > 0.\]

In summary, the zero line of $s_2$ must cross the square $\Sigma$ from left to right, while the zero line of $c_2$ crosses the square from top to bottom. Therefore an intersection between both lines must exist in $\Sigma$, corresponding to a choice of $\nu_3$ and $\nu_4$ where $f_2$ meets the Fourier condition of Lemma \ref{lem:lower_bounb_variation_f}.

\subsection{Constructing $f_3$}

Based on $f_2$ as constructed above with $\nu_3$ and $\nu_4$ chosen such that the Fourier condition holds, define the simplified step function 
\begin{equation} \label{eq:def_f3_new}
    \begin{aligned}
         f_3(t) :=& \;\nu_5 f_2(l)\chi_{(\tau_1,\nu_6]}(t) + \nu_5f_2(r) \chi_{[\nu_6,\tau_2)}(t)\\
        &+ (\delta + \nu) \chi_{\frac{\pi}{2}}(t) + \delta \chi_{(\frac{\pi}{2}, \pi]}(t)
    \end{aligned}
\end{equation}
%\begin{equation} \label{eq:def_f3_new}
%    \begin{aligned}
%         f_3(t) &:= \nu_5 \nu_3 \left(\sup_{t'\in[\tau_1,t_x]} f_1(t')\right)\chi_{(\tau_1,\nu_6]}(t) \\
%        &+ \nu_5\nu_4\left(\inf_{t'\in[t_x, \tau_2]} f_1(t')\right) \chi_{[\nu_6,\tau_2)}(t) \nonumber\\
%        &+ (\delta + \nu) \chi_{\frac{\pi}{2}}(t) + \delta \chi_{(\frac{\pi}{2}, \pi]}(t)
%\end{aligned}
%\end{equation}
for $t \in (0,\pi]$ and $f_3(t) = - f_3(t-\pi)$ for $t \in (\pi, 2\pi]$. Here 
$\nu_5 \in [0,1]$ and $\nu_6 \in [l,r]$ are two initially unknown parameters (cf. Figure \ref{fig:plot_f2_f3}).
It is obvious that the total variation has not increased, \ite, 
\[V(f_3) \le V(f_2).\]

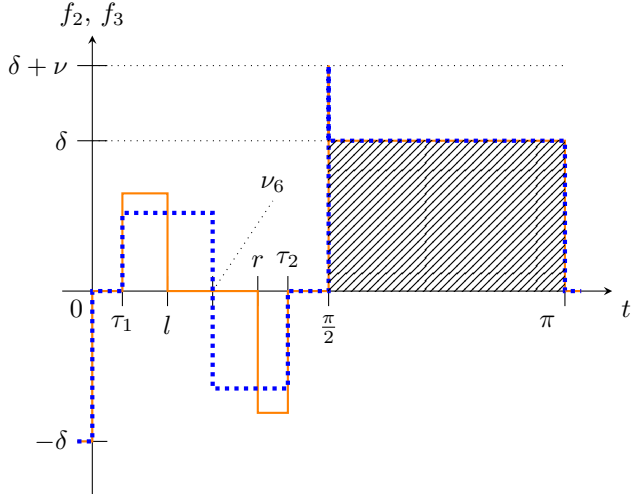
\begin{figure}
    \centering
    \begin{tikzpicture}[>=stealth, scale = 2]

\def\l{0.5}
\def\r{1.1}
\def\miu{1.5}
\def\nufive{0.8}
\def\nusix{0.8}

% Draw the axes
\draw[->] (0,-1.35) -- (0, 1.7) node[above] {$f_2$, $f_3$};
\draw[->] (-0.2,0) -- ({1.1*pi},0) node[below right] {$t$};

% horizontal line for mu, alpha
\draw[dotted] (0 , \miu) -- (3.15 ,\miu);
\draw[dotted] (0 , 1) -- (3.15 , 1);

% axis labels
\draw (0,0) node[below left] {$0$};
\draw (0.2, 0.1) -- ( 0.2,-0.1) node[below] {$\tau_1$};
\draw (\l,  0.1) -- (  \l,-0.1) node[below] {$l$};
\draw (\nusix,  0.1) -- (  \nusix,-0.1);
\draw [dotted] (\nusix,  0.0) -- (0.4 + \nusix, 0.6) node[above] {$\nu_6$};
\draw (\r, -0.1) -- (  \r, 0.1) node[above] {$r$};
\draw (1.3,-0.1) -- ( 1.3, 0.1) node[above] {$\tau_2$};
\draw ({pi/2},0) -- ({pi/2},-0.1) node[below] {$\frac{\pi}{2}$};
\draw ({pi},  0) -- ({pi}  ,-0.1) node[below left] {$\pi$};
%\draw (0,1) node[left] {$\delta$};
% Add tick mark at y=1
\draw (0.1,  \miu) -- (-0.1,  \miu) node[left] {$\delta+\nu$};
\draw (0.1,  1) -- (-0.1,  1) node[left] {$\delta$};
\draw (0.1, -1) -- (-0.1, -1) node[left] {$-\delta$};

% Function f_2
\draw[\orange, thick]
  plot coordinates {
    (-0.1,     -1.0)
    ( 0.0,     -1.0)
    ( 0.0,      0.0)
    ( 0.2,      0.0)
    ( 0.2,      1.3 * 0.5)
    ( \l,       1.3 * 0.5)
    ( \l,       0.0)
    ( \r,       0.0)
    ( \r,      -0.9 * 0.9)
    ( 1.3,     -0.9 * 0.9)
    ( 1.3,      0.0)
    ( {pi/2},   0.0)
    ( {pi/2},   \miu)
    ( {pi/2},   1)
    ( {pi},     1.0)
    ( {pi},     0.0)
    ( 3.25,     0.0)
  };

% Function f_3
\draw[\blue,dotted, ultra thick]
  plot coordinates {
    (-0.1,     -1.0)
    ( 0.0,     -1.0)
    ( 0.0,      0.0)
    ( 0.2,      0.0)
    ( 0.2,      1.3 * 0.5 * \nufive)
    ( \nusix,   1.3 * 0.5 * \nufive)
    ( \nusix,   -0.9 * 0.9 * \nufive)
    ( 1.3,     -0.9 * 0.9 * \nufive)
    ( 1.3,      0.0)
    ( {pi/2},   0.0)
    ( {pi/2},   \miu)
    ( {pi/2},   1)
    ( {pi},     1.0)
    ( {pi},     0.0)
    ( 3.25,     0.0)
  };

% 5) Fill the area between the horizontal line f=alpha (for t in [pi/2, pi])
%    and the t-axis (f=0) with a hatching pattern 

\fill[pattern=north east lines,draw=none]
  ({pi/2},0) -- ({pi},0) -- ({pi},1) -- ({pi/2},1) -- cycle;

\end{tikzpicture}
    \caption{The step function $f_3$ \ifcolor (blue dotted line)\else (dotted line)\fi  is constructed from $f_2$ \ifcolor (orange line)\else (solid line)\fi, here shown for $\nu_5 = 0.8$ and $ \nu_6 = \frac{r+l}2$.}
    \label{fig:plot_f2_f3}
\end{figure}

To conclude step 3 it remains to show that $\nu_5$ and $\nu_6$ can be chosen such that $f_3$, just like $\hat f, f_1$ and $f_2$, satisfies the Fourier condition of Lemma \ref{lem:lower_bounb_variation_f}, \ite, that the integrals with $\sin t$ and $\cos t$ vanish. Once more, the proof is similar to the corresponding one for $f_1$ and $f_2$:

\begin{figure}
    \centering
    \begin{tikzpicture}[>=stealth, scale = 2]

\def\e{0.0}  % distance between both functions

% Draw the axes
\draw[->] (0,   -0.3) -- (0,  1.7) node[left] {$\nu_5$};
\draw[->] (-0.5,   0) -- (2.0,  0) node[below] {$\nu_6$};

% axis labels
\draw (0,0) node[below left] {$0$};
\draw (0.1,   1) -- (-0.1,   1) node[left] {$1$};
\draw (1.3, 0.1) -- ( 1.3,-0.1) node[below] {$r$};
\draw (0.3, 0.1) -- ( 0.3,-0.1) node[below] {$l$};
\draw (0.8, 0.1) -- ( 0.8,-0.1) node[below] {$t_a$};

% bottom line
\draw[black, ultra thick]
  plot coordinates {
    ( 1.3,      0.0)
    ( 0.3,      0.0)
  };

\draw [dotted] (0.3, 0.0) -- (0.3, 1.0);
\draw [dotted] (0.8, 0.0) -- (0.8, 1.0);
\draw [dotted] (1.3, 0.0) -- (1.3, 1.0);
\draw [dotted] (0.0, 1.0) -- (1.3, 1.0);

% three dots
\draw[black, ultra thick]
  plot coordinates {( 0.3,      0.95) ( 0.3,      1.05)};
\draw[black, ultra thick]
  plot coordinates {( 0.25,      1.0) ( 0.35,      1.0)};

\draw[black, ultra thick]
  plot coordinates {( 0.8,      0.95) ( 0.8,      1.05)};
\draw[black, ultra thick]
  plot coordinates {( 0.75,      1.0) ( 0.85,      1.0)};

\draw[black, ultra thick]
  plot coordinates {( 1.3,      0.95) ( 1.3,      1.05)};
\draw[black, ultra thick]
  plot coordinates {( 1.25,      1.0) ( 1.35,      1.0)};

\draw (0.3, 0.0) node[above right] {$\Sigma_3$};

% c_3 = 0 line
\draw[\orange, thick] plot[smooth] coordinates{(0.4, 1.0) (0.6, 0.3) ( 1.3,   0.2)} node[right] {$c_3 = 0$};
  
% s_3 = 0 line
\draw[\blue, thick] plot[smooth] coordinates{(1.2, 1.0) (0.8, 0.5) (0.3, 0.3)};

\draw [dotted] (-0.1, 0.3) -- (0.3, 0.3);
\draw[\blue] (-0.1, 0.3) node[left] {$s_3 = 0$};

%  s, c vs. 0
% bottom line label
\draw [dotted] (1.0, 0.0) -- (1.1, -0.3);
\draw (1.1, -0.3) node[below right] {\shortstack{\textcolor{\blue}{$s_3 > 0$} \\ \textcolor{\orange}{$c_3 < 0$}}};

% left point
\draw [dotted] (0.3, 1.0) -- (-0.2, 1.2);
\draw (-0.2, 1.1) node[above left] {\shortstack{\textcolor{\blue}{$s_3 < 0$} \\ \textcolor{\orange}{$c_3 < 0$}}};

% center point
\draw [dotted] (0.8, 1.0) -- (0.9, 1.3);
\draw (0.9, 1.3) node[above] {\shortstack{\textcolor{\blue}{$s_3 < 0$} \\ \textcolor{\orange}{$c_3 > 0$}}};

% right point
\draw [dotted] (1.3, 1.0) -- (1.5, 1.1);
\draw (1.5, 1.1) node[above right] {\shortstack{\textcolor{\blue}{$s_3 > 0$} \\ \textcolor{\orange}{$c_3 > 0$}}};

\end{tikzpicture}
    \caption{Estimating the values of $s_3(\nu_5,\nu_6)$ and $c_3(\nu_5,\nu_6)$ at three points and at the $\nu_5 = 0$ line shows that the zero sets of $s_3$ and $c_3$ cross inside the square. The intersection corresponds to a choice of $\nu_5$ and $\nu_6$ which `balances' $f_3$ with respect to the Fourier condition.}
    \label{fig:plot_square_f3}
\end{figure}
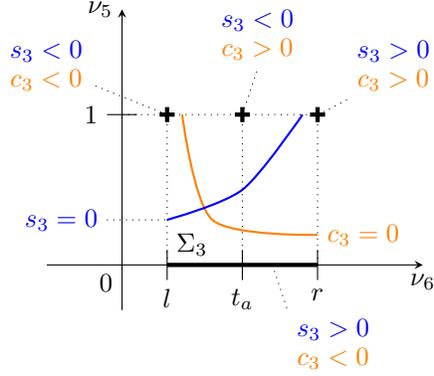 

The parameter space of $\nu_5$ and $\nu_6$ can be identified with the rectangle $\Sigma_3=[0,1] \times [l,r]$ in $\R^2$. Let
\begin{equation*}
    s_3(\nu_5, \nu_6) := \int_\Omega f_3(t) \sin t \dt 
\end{equation*}
and
\begin{equation*}
    c_3(\nu_5, \nu_6) := \int_\Omega f_3(t) \cos t \dt.
\end{equation*}

Once again, analyze the signs of $s_3$ and $c_3$ at certain special points (cf. Figure \ref{fig:plot_square_f3}):

Bottom line: For $\nu_5 = 0$ the only contribution to $s_3$ and $c_3$ comes from the terms $\delta \chi_{(\frac{\pi}{2}, \pi]}(t)$ and, by anti-periodicity, $-\delta \chi_{(\frac{3\pi}{2}, 2\pi]}(t)$ in (\ref{eq:def_f3_new}). Thus $s_3(0,\nu_6) > 0$ and $c_3(0,\nu_6) < 0$.

Upper corners: For $\nu_5 = 1$, the functions $f_2$ and $f_3$ are identical except for the interval $(l, r)$ and its anti-periodic counterpart $(l+\pi, r+\pi)$. For $\nu_6 = r$, this difference adds a positive quantity to the integrals over sine and cosine, and for $\nu_6 = l$ a negative one. Therefore
\[s_3(1,l) <  0 < s_3(1,r)\]
and
\[c_3(1,l) <  0 < c_3(1,r)\]

Third point: Finally, set $\nu_5 = 1$ and $\nu_6 = t_a$ where $t_a$ is chosen such that 
\[\int_{[l, t_a]} f_3(t) \dt = - \int_{[t_a, r]} f_3(t) \dt.\]
Then, compared to $f_2$, function $f_3$ gets the same additional area above the $t$ axis as below the axis on the interval $[l,r]$.
Therefore the monotonicity properties of sine and cosine imply that
\[s_3(1,t_a) <  0 < c_3(1,t_a).\]

Unlike $s_2$ and $c_2$, the functions $s_3$ and $c_3$ are not linear (in $\nu_6$) and therefore their zero sets are in general not straight lines. Nevertheless, since $s_3$ and $c_3$ are continuous and strictly monotonous in $\nu_5$ and $\nu_6$, respectively, their zero sets are strictly monotonous $C^1$ curves by the Implicit Function Theorem. Figure \ref{fig:plot_square_f3} shows how the zero set of $s_3$ connects the left side of the rectangle $\Sigma_3$ to the $[t_a, r]$ interval on the top side. The zero set of $c_3$ connects the right side of the rectangle to the $[l, t_a]$ interval. An intersection between both lines must exist in $\Sigma$, corresponding to a choice of $\nu_5$ and $\nu_6$ where $f_3$ meets the Fourier condition of Lemma \ref{lem:lower_bounb_variation_f}.
% Implicit function theorem: https://arxiv.org/pdf/1212.2066

\subsection{Finding $f_4$}

Lastly, it will be shown that $f_3$ as constructed in the previous step belongs to a class of functions, one of which minimizes the total variation to the bound claimed in the lemma.

Consider the class $F$ of functions $\phi: \Omega\rightarrow\R$ that can be written as
\begin{equation} \label{eq:def_f3}
    \phi(t) = \beta \chi_{(\tau_1,m]}(t) - \gamma\chi_{[m,\tau_2)}(t) +(\delta + \nu) \chi_{\frac{\pi}{2}}(t) + \delta \chi_{[\frac{\pi}{2}, \pi]}(t)
\end{equation}
for $t \in (0,\pi]$ and $\phi(t) = - \phi(t-\pi)$ for $t \in (\pi, 2\pi]$, with the parameters $\tau_1 < m < \tau_2$ and $\beta>0$ and $\gamma>0$ such that $\phi(t)$ meets the Fourier condition.

Then $f_3 \in F$ and the remaining task is to identify the particular function $f_4 \in F$ which minimizes the total variation.

Adding up the step changes along $\Omega$, one can establish that 
\begin{equation} \label{eq:total_var_f4}
    V(f_4) = 4(\delta + \nu + \beta + \gamma).
\end{equation}
With $\delta$ and $\nu$ being constant, minimizing $V(f_4)$ amounts to finding the minimal possible sum
\begin{equation*}
    S(m) := \beta + \gamma \quad \text{for } \tau_1 < m < \tau_2,
\end{equation*}
where $\beta$ and $\gamma$ depend on $m$ via the Fourier conditions
\begin{align*}
    0 &= \frac 12 \int_\Omega f_4(t) \sin t \dt\\
    &=  \beta (\cos\tau_1- \cos m) + \gamma(\cos \tau_2 -\cos m) + \delta
\end{align*}
and
\begin{align*}
    0 &= \frac 12 \int_\Omega f_4(t) \cos t \dt \\
    &= \beta (\sin m-\sin\tau_1) -\gamma(\sin\tau_2 - \sin m) - \delta.
\end{align*}
These linear equations define $\beta$ and $\gamma$ for every choice of $m$. Solving the first one for $\beta$ and the second one for $\gamma$, one obtains
\begin{eqnarray}
    \beta &=& \frac{\gamma(\cos m-\cos\tau_2) - \delta}{\cos\tau_1- \cos m}, \label{eq:beta_new}\\
    \gamma &=& \frac{\beta (\sin m-\sin\tau_1)- \delta}{\sin\tau_2 - \sin m}.\label{eq:gamma_new}
\end{eqnarray}
Putting $\gamma$ from (\ref{eq:gamma_new}) into (\ref{eq:beta_new}) and vice versa, one can express $\beta$ and $\gamma$ directly in terms of $l$ and $r$. Adding them yields an explicit expression for $S(m)$:
\begin{equation} \label{eq:def_of_S}
     S(m) = \delta\frac{\sin\tau_2-\sin\tau_1+\cos\tau_1-\cos\tau_2}{ \sin(\tau_2-m) + \sin(m-\tau_1) - \sin(\tau_2-\tau_1)}
\end{equation}
A standard search for extreme values then shows that $S(m)$ is minimal when
\[m = \frac{\tau_1+\tau_2}{2}.\]
Putting this into (\ref{eq:def_of_S}) and applying standard trigonometric identities, the lower bound on $S(m)$ can be simplified to
\begin{comment}
    \begin{align*} 
         S(m) &= \delta\frac{\sin\tau_2-\sin\tau_1+\cos\tau_1-\cos\tau_2}{ \sin(\frac{\tau_2-\tau_1}2) + \sin(\frac{\tau_2-\tau_1}2) - \sin(\tau_2-\tau_1)}\\
         &= \delta\frac{-\sqrt 2 \cos(\tau_2 + \frac\pi 4) + \sqrt 2 \cos(\tau_1 + \frac\pi 4)}{ 2\sin(\frac{\tau_2-\tau_1}2) - \sin(\tau_2-\tau_1)}\\
         &= 2\sqrt 2 \delta\frac{\sin(\frac{(\tau_1+\tau_2)}2+\frac\pi 4)\sin(\frac{(\tau_2-\tau_1)}2)}{ 2\sin(\frac{\tau_2-\tau_1}2) - \sin(\tau_2-\tau_1)}\\
         &= 2\sqrt 2 \delta\frac{\sin(\frac{(\tau_1+\tau_2)}2+\frac\pi 4)\sin(\frac{(\tau_2-\tau_1)}2)}{ 2\sin(\frac{\tau_2-\tau_1}2) - 2\sin(\frac{\tau_2-\tau_1)}2\cos(\frac{\tau_2-\tau_1)}2}\\
        &= \sqrt 2 \delta\frac{\sin(\frac{(\tau_1+\tau_2)}2+\frac\pi 4)}{ 1 - \cos(\frac{\tau_2-\tau_1)}2)}\\
        &= \sqrt 2 \delta\frac{\sin(\frac{(\tau_1+\tau_2)}2+\frac\pi 4)}{ 2 \sin^2(\frac{\tau_2-\tau_1}4)}\\
    \end{align*}
\end{comment}
\begin{equation*} %\label{eq:S_min_simplified}
    \min_{\tau_1 < m < \tau_2} S(m)  = \sqrt 2 \delta\,\frac{\sin\left(\frac{\tau_1+\tau_2}{2}+\frac\pi4\right)}{2\sin^2\frac{\tau_2-\tau_1}{4}}.
    %\min_{\tau_1 < m < \tau_2} S(m)  = \delta\,\frac{\cos\tau_1-\sin\tau_1+\sin\tau_2-\cos\tau_2}{\cos(\tau_2-\tau_1)-1}.
\end{equation*}
According to (\ref{eq:total_var_f4}) the minimal total variation of $f_4$ is then
\begin{equation} \label{eq:V4_min}
    V(f_4) = 4(\delta + \nu) + 2\sqrt 2 \delta\;\frac{\sin\left(\frac{\tau_1+\tau_2}{2}+\frac\pi4\right)}{\sin^2\frac{\tau_2-\tau_1}{4}}.
\end{equation}
Since $V(\hat f) > V(f_4)$, this concludes the proof of Lemma \ref{lem:lower_bounb_variation_f}.    

\begin{table}[h!] 
    \centering
    \begin{tabularx}{\linewidth}{|X|X|} 
        \hline
        \textbf{Assumption} & \textbf{Justification} \\ \hline
        $\Gamma$ of length $2\pi$ & The Ovals Conjecture is homogeneous with respect to dilation and thus $\lambda_\Gamma$ scales in a predictable way when changing the length of the curve\\ \hline
        $\Gamma$ planar, strictly convex and smooth & Minimizers of the loop problem meet these conditions as proven in \cite{Denzler:Existence}  \\ \hline
        $\phi' > 0$ & Invariance under reflection of $\Gamma$ \\ \hline
        $\phi(0) = 0$ & Invariance under re-parametrization of $\Gamma$ \\ \hline
        $t_\lambda = 0$, \ite $I(0) = I(\frac{\pi}{2}) = \lambda_\Gamma$ & Invariance under rotation of $\Gamma$, determined only up to rotations by $\frac\pi 2$ \\ \hline
        $I(t) > \lambda_\Gamma$ for $t \downarrow 0$, \ite all critical angles in $(0,\frac{\pi}{2}) \cup (\pi,3\frac{\pi}{2})$ & Further rotation of $\Gamma$ by $0$ or by $\frac\pi 2$, \ite rotation is now determined up to an angle of $\pi$ \\ \hline
        $f(t) > 0$ on $[\frac{\pi}{2},\pi]$ & Further rotation of $\Gamma$ by $0$ or $\pi$\\ \hline
    \end{tabularx}
    \caption{Assumptions on the curve $\Gamma$ and its parametrization.}
    \label{tab:assumptions}
\end{table}

\printbibliography 

\end{document}